\documentclass[11pt]{article}
\usepackage{amsmath,amssymb,amsthm,mathrsfs}
\usepackage{graphicx}
\usepackage{xcolor}
\usepackage{booktabs}
\usepackage{multirow}
\usepackage[authoryear,round]{natbib}
\usepackage{hyperref}
\hypersetup{
  colorlinks=true,
  linkcolor=blue,
  citecolor=blue,
  urlcolor=blue
}

\newtheorem{structure}{Structure}

\theoremstyle{definition}
\newtheorem{lemma}{Lemma}[section]
\newtheorem{proposition}{Proposition}[section]

\newtheorem{definition}{Definition}[section]

\newtheorem{example}{Example}[section]
\newtheorem{corollary}{Corollary}[section]
\newtheorem{remark}{Remark}[section]

\newcommand{\EE}{\mathbb{E}}
\newcommand{\RR}{\mathbb{R}}

\newcommand{\FF}{\mathbb{F}}

\newcommand{\VV}{\mathbb{V}}

\newcommand{\pr}{\mathop{\mathrm{pr}}}

\usepackage{relsize}
\newcommand{\T}{{\mathsmaller {\rm T}}}
\newcommand{\ind}{{\rm I\hspace{-2.2mm} 1}}
\newcommand{\cond}{{\rm \,| \,}}

\newcommand{\argmax}{\mathop{\mathrm{argmax}}}

\usepackage[T2A,T1]{fontenc}
\newcommand{\ya}{\mbox{\usefont{T2A}{\rmdefault}{m}{n}\cyrya}}

\title{Induced replication and the assessment of models}
\author{
  Heather Battey\thanks{Department of Mathematics, Imperial College London, UK. Email: \texttt{h.battey@imperial.ac.uk}.} \quad
  and \quad
  Nancy Reid\thanks{Department of Statistical Sciences, University of Toronto, Canada.}
}
\date{\today}

\begin{document}

\maketitle

\begin{abstract}
		The paper is a first attempt at a foundation of inference for the assessment of semiparametric and other high-dimensional models, seeking reconciliation with the low-dimensional Fisherian foundations by replacing out-of-sample prediction error as a criterion for semiparametric model assessment by a type of within-sample prediction error. We highlight the possibility, in some contexts, of avoiding the usual semiparametric considerations, including estimation of unparametrised nuisance components and the associated difficulties of tuning-parameter choice that blur the distinction between estimation and model assessment. A key aspect of the work is the inducement of replication under the postulated model, the route to which being context-specific but based on structural properties of the model. The ideas are isolated through examples.

	\bigskip
	
	\noindent \emph{Some key words:} co-sufficiency; 
	exchangeability; 
	inferential separations;
	model adequacy;
	post-reduction inference;
	semiparametic models.
\end{abstract}

\section{Introduction}\label{secIntro}

A commonly held view, popularised by \citet{Popper} in equivalent terminology, is that a model should be assessed by its ability to predict out of sample. The perspective of the present work is that structure within models, semiparametric or otherwise, allows model assessment through a type of predictive performance within sample. This perspective, emanating from \citet{Bartlett1937}, emerges clearly in idealised parametric contexts, and can be extended to certain types of semiparametric model through reparametrisation and hitherto-unexplored inferential separations. 

Even in favourable parametric contexts, where the information for model assessment separates sharply from that relevant for parameter inference, \citeauthor{Bartlett1937}'s (\citeyear{Bartlett1937}) formulation is detached from practicalities, and may have little sensitivity to departures from the postulated model. Consideration of various model structures has led us to a unified formulation involving the inducement of internal replication: multiple abstract partitions of the sample, within which data may supply evidence against a false model. If this evidence is cumulatively large over partitions, then that casts doubt on the adequacy of the model. While this has something of the flavour of cross-validation, an insight of the present work is that the model structure largely dictates the partition, which is typically not a natural partition by sampling units.

These rather abstract concepts can be clarified first through idealised examples in which replication can be induced exactly, and then in a more ambitious setting, representative of typical statistical formulations. To foreshadow these more challenging cases, consider semiparametric regression in an exponential-family model. The expected value of the outcome $Y_i$ conditional on a binary treatment indicator $t_i$ and a scalar covariate $x_i$ has the form
\begin{equation}\label{eqGeneralForm}
g(\EE(Y_i \mid x_i, t_i)) = h(x_i) + \beta t_i,
\end{equation}
where $h$ is an unknown function and $g=(\kappa')^{-1}$ is the canonical link function of the exponential family, with cumulant generating function $K_i(u; \beta, h)=\kappa(t_i\beta + h(x_i) + u) - \kappa(t_i\beta + h(x_i))$. Model \eqref{eqGeneralForm} can be viewed as a version of the exponential-family semiparametric single-index model of \citet{Carroll1997}. A typical semiparametric analysis would impose smoothness assumptions on $h$, implying an approximating basis of order smaller than the sample size $n$. The smoothness might be imposed by using a truncated basis expansion, by regularisation, or by kernel smoothing, with tuning parameters chosen by cross-validation, blurring the distinction between estimation and model assessment in the data at hand. Any assessment of model adequacy through that route then seems to require an additional sample on which prediction accuracy can be assessed.

An insight of this work is that many semiparametric models can be reparametrised to circumvent estimation of the nonparametric component $h$ by exploiting inferential separations within the postulated model. The separations also allow an inducement of replication of known form if and only if the postulated model is adequate. Through this route, the data are given the opportunity to refute the model at some chosen significance level, a statistical analogue of proof by contradiction.

In the same way that the parametric foundations apply seamlessly to some, but not all, types of parametric model, it is inconceivable that every semiparametric model can be covered by an identical formulation. Thus, while there are unifying concepts underpinning the work, a general methodology with uniformly-valid error rates is out of scope. We have sought instead to isolate insight and foundational understanding through examples.

\section{Model assessment as within-sample prediction}\label{secWithin}

\subsection{Ancillary and co-sufficient variation}

Two widely used approaches to assessing regression models are to substitute unknown regression parameters by estimates and check residuals for any anomalous behaviour, or to assess the predictive performance of the fitted model. The visually compelling but sometimes informal approaches based on residuals can be viewed as a way of evaluating predictive accuracy within sample. The two-stage exposition of \citet[][p.~29]{BNC1994} for more general parametric models clarifies this claim: the statistician first learns the value $S=s^o$ of the sufficient statistic for a parameter of a given statistical model, before learning the additional information that the outcome is $Y=y^o$. If $y^o$ is extreme when calibrated against the conditional distribution of $Y$ given $S=s^o$, that casts doubt on the adequacy of the model. The conditional distribution is a type of within-sample predictive distribution. The formulation was first proposed by \citet{Bartlett1937}, although he did not use the above analogy to prediction. The term \emph{co-sufficient sampling} was introduced by \citet{Lockhart2009} and adopted in subsequent work by \citet{BJ2022}, who substantially advanced the methodology and associated theory for simulation from the conditional distribution. The present paper does not involve simulation and is based instead on several types of inferential separation, some non-standard.

Following the geometric exposition of \citet{Battey2024}, it is convenient to write the information in $Y$ as $(S,Q(S))$. In this notional idealised separation, $Q(S)$ is not a function of $S$ in the conventional sense, and typically does not have an analytic form in terms of the observations; rather $Q(s^o)$ has the conditional distribution of $Y$ given $S=s^o$. The realisation $s^o\in\RR^d$ fixes the co-sufficient manifold 
\begin{equation}\label{eqCoSuffMan}
\mathcal{Q}(s^o)=\{y\in\RR^n: s(y)=s^o\}\subset \RR^n,
\end{equation}
leaving $n-d$ degrees of freedom for variation of $y$ consistent with the constraint $s(y)=s^o$, i.e., $\mathcal{Q}(s^o)$ is a manifold of dimension $n-d$ embedded in $\RR^n$. Once $s^o$ is observed, the notional random variable $Q(s^o)$ captures the \emph{co-sufficient variation}: the residual randomness in $Y$ after conditioning on $S=s^o$, which is constrained to $\mathcal{Q}(s^o)$. The implicit general definition of $Q(s^o)$ is
\begin{equation}\label{eqCosuffStat}
\pr(Q(s^o)\in \mathcal{A})=\pr(Y\in \mathcal{A}|S=s^o),
\end{equation}
where $\mathcal{A}\subseteq \mathcal{Q}(s^o)$ by definition. 

The information is $Q(s^o)$ is not, in general, the only information available for model assessment. When the dimension of $S$ is larger than the dimension of the unknown parameter $\theta$.  Write $S\cong (T,A)$ where $T$ estimates $\theta$ and, from observation of $A$ alone, it is difficult or impossible to estimate $\theta$. We take this property as our definition of ancillarity. Ideally, $A$ is ancillary in the strongest sense, having a distribution that does not depend on $\theta$, which implies the factorisation
\begin{equation}\label{eqAncillFact}
f_{S}(s;\theta)= f_{S\mid A}(s \mid a;\theta)f_A(a).
\end{equation}
Since the information relevant for inference for $\theta$ is contained in the conditional distribution of $S$ given $A=a^o$, calibration of $a^o$ against the marginal distribution of $A$ is a second route to model assessment. 

Whether the co-sufficient or ancillary information is used, the logic is that of proof by contradiction: suppose for a potential contradiction at level $\alpha$ that the true distribution belongs to model $\mathcal{M}_0$. The structure of $\mathcal{M}_0$ implies inferential separations that are violated when the model assumption is false. One would expect, therefore, the realisation $y^o$ or $a^o$ to be extreme, as quantified by the chosen significance level $\alpha$, when calibrated against the co-sufficient or ancillary distributions of $Q(s^o)$ or $A$. Any such observed extremity indicates that the data is not compatible with $\mathcal{M}_0$ at level $\alpha$.

A perspective of this paper is that different semiparametric model structures sometimes have inferential separations useful for model assessment, complementing those above. A first generalisation of co-sufficient variation replaces $S$ by a parameter-dependent statistic $S(\psi)$, sufficient for a nonparametric component if and only if the true distribution of $Y$ is in $\mathcal{M}_0$ and specified by $\psi=\psi^*$, the value of a low-dimensional parameter. By definition, a semiparametric model has a low-dimensional interest parameter and a nonparametric nuisance component; \eqref{eqGeneralForm} is an example of such a model. We use that terminology here, even though all parameters can be viewed as nuisance parameters for the purpose of model assessment. A further extension allows $S$ or $S(\psi^*)$ to be conditionally sufficient, conditional on a component of the data, such as a failure having occurred at a particular time point, in the case of the proportional hazards model in Example \ref{examplePH} below. 

An analogue of ancillary-based model assessment involves a random variable $A(\psi)$ that is ancillary for the nonparametric component at some value $\psi=\psi^*$ if and only if the true distribution of $Y$ is in $\mathcal{M}_0$. For the proportional hazards model, the natural ancillary statistic for the baseline hazard function is conditionally sufficient for the regression parameter, showing that two routes to model assessment can agree in some contexts.

\subsection{A parametric example}\label{secLinearReg}

Normal-theory linear regression provides a transparent illustration of the sufficiency separation and clearly delineates the information available for model assessment within sample.

\begin{example}\label{exampleRegression}
	Consider the regression model $Y\sim N(\mu, \sigma^2 I_n)$, where $\mu=X\beta \in \mathcal{X} \subset \RR^n$, with $\mathcal{X}$ the column space of $X$, full rank by assumption, and $\beta\in \RR^{d_\beta}$. The residual vector $\hat\varepsilon$ belongs to the orthogonal complement $\mathcal{X}^\perp\subset \RR^n$, which has dimension $n-d_\beta$. With $y^o\in \RR^n$ the observed outcome and $y\in\RR^n$ an arbitrary value, the $n$ degrees of freedom for variation of $y$ separate, first according to $\mathcal{X}\oplus \mathcal{X}^\perp = \RR^n$ and subsequently within $\mathcal{X}^\perp$, where one degree of freedom is used for estimation of $\sigma^2$, leaving $n-(d_\beta + 1)$ degrees of freedom for assessment of the mean model $\mu \in\mathcal{X}$. 
	
	From a sufficient statistic, $S=s(Y)=(X^\T Y, \hat\varepsilon ^\T \hat\varepsilon)$ say, the information remaining after conditioning on the observed value $s^o$ of $S$ is referred to as \emph{co-sufficient}. In general there need not exist an explicit statistic that carries all the co-sufficient information without loss or redundancy, but in the present example, a co-sufficient statistic is available. 
	
	Let $V=(v_1,\ldots, v_{d_\beta})$ be the matrix of eigenvectors corresponding to the unit eigenvalues of the projection matrix $I-X(X^\T X)^{-1} X^\T$. The projection $V V ^\T Y\in\mathcal{X}^\perp$ has coordinates $V ^\T Y$ in the basis $V$. A statistic independent of the squared length $\hat\varepsilon^\T \hat\varepsilon = \|V^\T Y\|^2$ is
	\begin{equation}\label{eqQ}
	Q = \frac{V^\T Y}{\|V^\T Y\|} \in \mathcal{S}^{n-d_\beta}, 
	\end{equation}
	where $\mathcal{S}^{n-d_\beta}$ is the unit hypersphere embedded in $\RR^{n-d_{\beta}}$. The aforementioned independence implies that $Q$ is unaffected by conditioning on $S=s^o$. Under the postulated mean model, $Q$ is uniformly distributed on $\mathcal{S}^{n-d_\beta}$, a conclusion that holds even if the normal distribution for $Y$ is replaced by a student $t$ distribution or other spherically symmetric distribution \citep[][Ch.~1.5]{Muirhead1982}. 
	
	There are $n-d_\beta - 1$ degrees of freedom for variation of the angles $a_1,\ldots, a_{n-d_\beta - 1}$, specifying the position of an arbitrary $q$ on $\mathcal{S}^{n-d_\beta}$. This shows how the $n$ degrees of freedom for variation in $y$ decompose, without loss or redundancy, into sufficient and co-sufficient components as
	\begin{align*}
	&	(X^\T y, \;\;\;\, \|V^\T y\|^2, \;\;\;\; a_{1},\ldots,a_{n-d_\beta - 1}) \\
	n\;\; =    & \;\;\; d_\beta \;\;+ \;\;\;\;\; 1 \;\;\;\;\;\; + \;\;\; n-(d_\beta + 1).
	\end{align*}
	\hfill$\square$
\end{example}

\subsection{Idealised semiparametric examples: two exact analyses}

Example \ref{exampleRegression}, in addition to highlighting the possibility of model assessment within sample, exposes the difficulties in semiparametric and other high-dimensional models. A genuinely infinite-dimensional component would leave no degrees of freedom for model assessment within sample. Smoothness assumptions implicitly constrain the degrees of freedom used in estimation, effectively reducing a semiparametric problem to a parametric one in which the information used for estimation is allowed to grow slowly with sample size. 

The following example makes two important points: that smoothness or other assumptions on infinite-dimensional nuisance parameters can sometimes be evaded, and that within-sample model assessment can be performed without estimating nuisance parameters and without sample splitting. 

\begin{example}[Semiparametric matched pair problem]\label{exIntro}
	Let $(Y_{j1}, Y_{j0})_{j=1}^m$ be outcomes on treated and untreated units from a matched-pair design. A general semiparametric specification allows each pair to have its own parameter $\gamma_j$, specifying the distribution at baseline, and a parameter $\psi$ quantifying a treatment effect that is stable across the pairs. This formulation is more general than a typical semiparametric one in the following sense. Had covariates $x\in\RR^p$ been available, the same model would have been valid with $\gamma_j = h(x_j)$ for a regression function $h: \RR^p \rightarrow \Gamma$ of arbitrary and unspecified form, conventionally belonging to a  smoothness class. 
	
	Suppose that $(Y_{j1}, Y_{j0})$ are independent and exponentially distributed with rates $\gamma_j\psi$, $\gamma_j$ respectively. In work primarily emphasising a different point, \citet[][\S 5]{BC2020} proposed an approach to model assessment for this example based implicitly on the following generalisation of the co-sufficiency argument. Suppose, for an analogue of a proof by contradiction, that the true distribution belongs to the postulated model and has true parameter value $\psi^*$, then $Y_{j0} +  Y_{j1}\psi^*=:S_{j}(\psi^*)$ is sufficient for $\gamma_j$ and has density function 
	\begin{equation}\label{gamma2lambda}
	f_{S_{j}(\psi^*)}(s) = \gamma_j^{2} s \exp(-\gamma_j s),
	\end{equation}
	i.e., $S_{j}(\psi^*)$ is gamma distributed with shape parameter 2 and rate parameter $\gamma_j$. The conditional density of $Y_{j1}$ at $y_{j1}$, given $S_{j}(\psi^*)=s_j(\psi^*)$, is $\psi^*/s_j(\psi^*)$. Equivalently $U_{j}(\psi^*):=Y_{j1}\psi^*/s_{j}(\psi^*)$ has a standard uniform distribution for all $j=1,\ldots,m$. See Figure \ref{figCosuff} for a depiction of the argument in terms of the co-sufficient manifold \eqref{eqCoSuffMan}. The above conclusions are invalidated if $\psi^*$ is replaced by any other value $\psi_0$, or if the model is inadequate, suggesting a route to assessment of model adequacy via the empirical behaviour of $U_{1}(\psi_0), \ldots, U_m(\psi_0)$ for a postulated value $\psi_0$, which can be viewed as analogues of residuals.
	
\begin{figure}				
	\begin{center}
		\includegraphics[trim=0in 0in 0in 0in, clip,width=.55\linewidth]{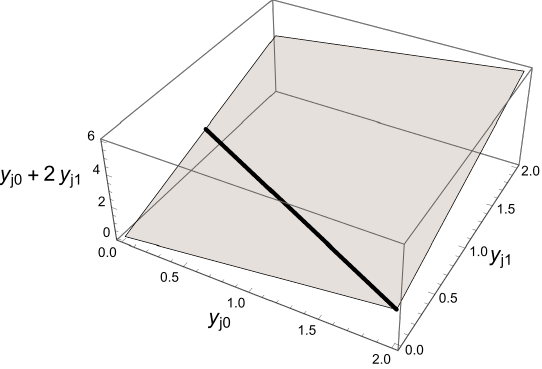} 
		
	\end{center}
{\vspace{-0.2cm}}
\caption{The surface is $s(\psi_0)=y_{j0}+\psi_0 y_{j1}$ as a function of $y_{j0}$ and $y_{j1}$ for the postulated value $\psi_0=2$. The black line is the co-sufficient manifold for the $j$th pair, $\mathcal{Q}_j(s^o)=\{(y_{j0},y_{j1}): y_{j0}+\psi_0 y_{j1} = s^o\}$ with $s^o=2$. The variate $U_j(\psi_0)$ is unifomly distributed on this line if and only if the model is correct and $\psi_0=\psi^*$. \label{figCosuff}}
\end{figure}	
	
	Some postulated values $\psi_0$ may produce a distribution for $U_1(\psi_0),\ldots, U_m(\psi_0)$ that is hard to distinguish from a standard uniform sample even when the model is violated. Unless $m$ is small, it is sensible to replace $\psi_0$ by the realisation of a  estimator $\hat\psi$, if such exists, that is consistent under the postulated model. The logic is that only one value of $\psi$ then needs to be refuted in order to refute the model, rather than all possible values. 
	
	In this example, the transformed random variables $Z_j=Y_{j1}/Y_{j0}$ are, under the model, identically distributed with  density function
	\begin{equation}\label{eqBC2020}
	f_Z(z; \psi^*) = \frac{\psi^{*}}{(1+\psi^* z)^2},
	\end{equation}
	free of all the pair-specific nuisance parameters. The resulting maximum likelihood estimator $\hat\psi$ of $\psi^*$ is therefore consistent as $m\rightarrow \infty$ provided that the model is adequate. An approach replacing $\psi^*$ by $\hat\psi$ in $U_1(\psi^*),\ldots,U_m(\psi^*)$ can be formally justified using standard arguments recorded in \S \ref{secEstimateParam}.   \hfill$\square$
\end{example}

Example \ref{exIntro} replaces the usual semiparametric approach of assessing prediction error after regularised estimation by an in-sample assessment based on the structure of the postulated model, avoiding, without sample splitting, any post-selection model inference issues that would arise from using the same sample to choose tuning parameters and to assess the model.

\begin{example}[Time-dependent Poisson process with semiparametric intensity function]\label{examplePoisson}
	The following is a semi-parametric generalisation of a model proposed by \citet[][]{Cox1955}. Suppose a time-dependent Poisson process has a semiparametric intensity function of regression form
	\[
	\lambda_i(t)=e^{h(x_i) + \beta t}
	\] 
	for the $i$th of $n$ individuals, where $h:\RR^d \rightarrow \RR$ is an unknown function of a covariate vector $x_i$, which may be unmeasured. The model could be suitable for sequences of health events, in which $h(x_i)$ captures individual-specific effects and $\beta$ captures the general effect of ageing. 
	
	To avoid assumptions on, and estimation of, $h$, write $\gamma_i=h(x_i)$. Suppose that in a fixed interval $(0,t_0)$, events are recorded at times $(t_{i1}, \ldots, t_{i m_i})$ for individual $i$. The likelihood contribution for the $i$th individual is 
	\begin{equation}\label{eqJointLik}
	L(\gamma_i, \beta; t_{i1}, \ldots, t_{i m_i}) = \exp\Bigl\{m_i\gamma_i + \beta \sum_{j=1}^{m_i}t_{ij} - e^{\gamma_i}(e^{\beta t_0} - 1)/\beta\Bigr\}.
	\end{equation}
	The jointly sufficient statistics for $(\gamma_i, \beta)$ based on the $i$th individual are $(M_i, \sum_{j=1}^{M_i}T_{ij})$, and the conditional distribution of $\sum_{j=1}^{M_i}T_{ij}$ given $M_i=m_i^o$ eliminates $\gamma_i$. There is, therefore, an induced replication in the conditional distributions across the $n$ individuals. The marginal probability mass function for the number of events for the $i$th individual is Poisson with rate
	\[
	\int_{0}^{t_0}\lambda_i(u)du = e^{\gamma_i}(e^{\beta t_0} - 1)/\beta.
	\]
	On using the right hand side of \eqref{eqJointLik}, viewed as a function of $t_{i1}, \ldots, t_{i m_i^o}$ for fixed parameter values, the conditional density is 
	\begin{align}\label{eqConditionalPDF}
	f( t_{i1}, \ldots, t_{i m_i^o} \mid m_i^o) = \; & \frac{m_i^o! \beta^{m_i^o} \exp(\beta \sum_{j=1}^{m_i^o}t_{ij})}{(e^{\beta t_0} - 1)^{m_i^o}},
	\end{align}
	which is the probability density function of an ordered sample of $m_i^o$ random variables each with density function
	\begin{equation}\label{eqIndivContributins}
	f_T(t;\beta)=\frac{\beta e^{\beta t}}{(e^{\beta t_0} - 1)}, \quad 0 \leq t \leq t_0.
	\end{equation}
	It follows that the conditional distribution of $S_i=\sum_{j=1}^{M_i}T_{ij}$ given $M_i=m_i^o$ is that of the sum of $m_i^o$ independent random variables with density function \eqref{eqIndivContributins}. Equation \eqref{eqIndivContributins} reduces to $1/t_0$ in the limit as $\beta\rightarrow 0$. \citet[][pp.~45--46]{CL1966}, who discussed a parametric analogue, did not use \eqref{eqIndivContributins} and suggested that a conditional likelihood based on \eqref{eqConditionalPDF} be used for inference on $\beta$. The most convenient route to the analogue of $U_j(\psi_0)$ in Example \ref{exIntro}, involves calculation of the conditional distribution of $S_i$ given $m_i^o$, from which we can construct the canonical transformation $U_i(\beta_0)=F_{S_i|M_i}(s \mid m_i^o ; \beta_0)$. As in Example \ref{exIntro}, an approach with more sensitivity uses $\hat \beta$ in place of $\beta_0$, where $\hat\beta$ maximises the conditional log likelihood function based on the product of joint conditional density functions \eqref{eqConditionalPDF} over the $n$ individuals.

	Appendix \ref{secLaplace} gives an exact expression for the conditional distribution function $	F_{S_i|M_i}(s \mid m_i^o ; \beta)$. The same function can, however, be estimated using Monte-Carlo replicates of sums of $m_i^o$ random variables drawn from the density function \eqref{eqIndivContributins}, and this was used in simulations in \S \ref{secNumericalMain}. For large $m_i^o$, a normal approximation to the distribution can be used. Up to Monte Carlo sampling error, the resulting estimates of $F_{S_i|M_i}(S_i \mid m_i^o ; \beta)$ follow a standard uniform distribution at the true value of $\beta$ under correct specification of the model, allowing a synthesis of evidence via Fisher's statistic as in Example \ref{exIntro}.
	
	A similar argument applies for versions of this model with intensity functions having a quadratic time trend $\lambda_i(t)=\exp(\gamma_i + \beta_1 t + \beta_2 t^2)$ or a power-law time trend $\lambda_i(t)=e^{\gamma_i}t^\beta$. These are semiparametric analogues of models discussed by \citet[][p.~138]{Cox1955}, where the essential elements of the above argument apply on noting that, in the first case, $M_i$, $S_{i1}=\sum_{j=1}^{M_i}T_i$ and $S_{i2}=\sum_{j=1}^{M_i}T_i^2$ are jointly sufficient in the $i$th sample for $\gamma_i$, $\beta_1$ and $\beta_2$, and in the second case, $M_i$ and $S_{i}=\sum_{j=1}^{M_i}\log T_i$ are jointly sufficient for $\gamma_i$ and $\beta$.
\end{example}

\section{Inducible replication}\label{secInduced}

Underpinning Examples \ref{exIntro} and \ref{examplePoisson} was the possibility of inducing an internal replication, the strongest form of which is given in Definition \ref{defReplInducement}. The term \emph{internal replication} refers to all types of replication of the forms described by Definitions \ref{defReplInducement} and \ref{defWeakReplInducement}, and also to the externally-induced internal replication of Definition \ref{defExtReplInducement} from \S \ref{secSynthetic}. With parametric information partitions in mind, we are led to constructions that retain a form of co-sufficient or ancillary information in the internal replicates.

\begin{definition}[\textbf{Induced exchangeability}]\label{defExchInducement}
	Let $Y_1,\ldots, Y_n$, each valued in $\FF$, be independent non-exchangeable outcomes under a postulated model $\mathcal{M}_0$. Suppose there exists a function or composition of functions $h_0: \FF^n \rightarrow \FF^m$, $1<m\leq n$, depending on $\mathcal{M}_0$ but not depending on any unknown parameters, such that the joint distribution of the induced random variables satisfies $(V_1,\ldots, V_m)\stackrel{d}{=}(V_{\pi(1)},\ldots, V_{\pi(m)})$ under $\mathcal{M}_0$ for any permutation $\pi:[m]\rightarrow [m]$, i.e.~$V_1,\ldots,V_m$ are exchangeable. Then we say that $h_0$ is \emph{exchangeability inducing} under $\mathcal{M}_0$. 
\end{definition}

\begin{definition}[\textbf{Induced replication}]\label{defReplInducement}
	If, in Definition \ref{defExchInducement}, $V_1,\ldots, V_m$ are independent and identically distributed, we say that $h_0$ is \emph{replication inducing} under the postulated model.
\end{definition}

To induce replication in Example \ref{exampleRegression}, the null marginal distribution function $F_{A_j}$ of each angle $A_j$ can be derived from the expression given by \citet[][p.~49]{Muirhead1982}. From these distribution functions, iid replicates $U_{1},\ldots, U_{n-d_\beta - 1}$, uniformly distributed in this case, can be constructed as $U_j = F_{A_j}(A_j)$, thus $h_0$ consists of the composition of transformations from the outcomes $Y$ to the angles $A_1,\ldots, A_{n-d_\beta - 1}$, and the $n-d_\beta - 1$ componentwise transformations $U_j = F_{A_j}(A_j)$. In this example, an analysis based on $Q$ from \eqref{eqQ}, or equivalently $U_{1},\ldots, U_{n-d_\beta - 1}$, has almost no power to detect an erroneous model unless a direction of departure has been specified. This is a separate point to which we return in Example \ref{exampleBRT} of \S \ref{secSynthetic}.

In Definition \ref{defReplInducement}, the identical distribution of the replicates may or may not depend on the unknown true value of a fixed low-dimensional parameter $\psi$ but it should not depend on the infinite-dimensional components of the model captured, for example, by the $\gamma_j=h(x_j)$ in Example \ref{exIntro}. Example \ref{exIntro} illustrates two forms of internal replication: one induced though operations based on a $\psi$-dependent sufficient statistic, and another induced by taking ratios $Z_{j}=Y_{j1}/Y_{j0}$ leading to a $\psi$-dependent distribution for the replicates. The more general Definition \ref{defWeakReplInducement} allows for the first type of replication inducement.

\begin{definition}[\textbf{Interest-dependent induced replication}]\label{defWeakReplInducement}
	If, in Definition \ref{defReplInducement}, $h_0$ depends on the true value $\psi^*$ of an unknown parameter $\psi$ of fixed finite dimension, then $h_0$ is said to be \emph{interest-dependent replication inducing} under the postulated model, giving rise to iid replicates $V_1(\psi^*), \ldots, V_m(\psi^*)$. 
\end{definition}

In Example \ref{exIntro}, the unknown value $\psi^*$ in Definition \ref{defWeakReplInducement} is replaced by an estimate, producing induced random variables $\hat V_1,\ldots, \hat V_m$ that are exchangeable but not independent. Provided that the estimator $\hat\psi$ is consistent for $\psi^*$ under the postulated model as some quantity $n\rightarrow \infty$, $\hat V_1,\ldots, \hat V_m$ will often be asymptotically independent under the same asymptotic regime.

In \S \ref{secExpFamilyReg}, we achieve a version of Definition \eqref{defWeakReplInducement} approximately, quantifying the approximation error in doing so, while in \S \ref{secSynthetic}, we generalise the above definitions to allow for externally-induced replication.

\section{Combination of information across replicates}\label{secFramework}

\subsection{Joint assessment of a model and its interest   parameters}\label{secFormulation}

The discussion here is based on Definition \ref{defWeakReplInducement}, easily adapted for the other forms of induced replication. For concreteness, and with no loss of generality, we take the induced replicates to be standard uniformly distributed when the postulated model is true with $\psi_0=\psi^*$, and therefore use the suggestive notation $U_1(\psi_0), \ldots, U_m(\psi_0)$. 

Suppose for an evidential analogue of proof by contradiction that the postulated model is true with parameter value $\psi_0$. The canonical way of combining the putatively iid standard-uniform replicates $U_1(\psi_0), \ldots, U_m(\psi_0)$ is via \citeauthor{Fisher1932}'s (\citeyear{Fisher1932}) statistic for combining $p$-values, traditionally used in meta analysis. With $G_{2m}$ the distribution function of a $\chi^2$ random variable with $2m$ degrees of freedom, the two-sided $p$-value is $2\min\{p(\psi_0), 1-p(\psi_0)\}$, where $p(\psi_0):=G_{2m}(-2\sum \log U_j(\psi_0))$. The resulting $\alpha$-level confidence set for the interest parameter $\psi$ under correct specification of the model is
\begin{equation}\label{eqConfSet}
\mathcal{C}(\alpha):=\Bigl\{\psi_0 \in \Psi: 2\min\{p(\psi_0), 1-p(\psi_0)\}>\alpha\Bigr\}.
\end{equation}
If the confidence set is empty at a chosen level $\alpha$, that casts doubt on the adequacy of the model. Section \ref{secEstimateParam} discusses the more powerful approach based on a null-consistent estimator, mirroring Example \ref{exIntro}.

The use of Fisher's statistic in \eqref{eqConfSet} is convenient for analytic calculations of power because the expectation and variance of $R_j = -\log U_j$ simplify, by integration by parts, to 
\begin{equation}\label{eqIntParts}
\EE(R_j)= \int_0^1 \frac{F_{U}(u)}{u}du, \quad
\VV(R_j) = -2\int_0^1 \frac{\log(u) F_{U}(u)}{u}du - \Bigl(\int_0^1 \frac{F_{U}(u)}{u}du\Bigr)^2,
\end{equation}
where $F_U$ is the distribution function of $U_j$. 

Both moments increase as the density function $f_U$ of $U$ concentrates near zero, and the mean exceeds half the variance, and hence departs from the null $\chi^2_2$ behaviour, when $f_U$ departs from uniformity asymmetrically towards zero. When departures instead occur towards 1, the mean and variance are typically both too small to be compatible with the null distribution, although in this case sensitivity can be increased by replacing $U_j(\psi_0)$ by $1-U_j(\psi_0)$ in $p(\psi_0)$ from \eqref{eqConfSet}. Thus, if the majority of departures from uniformity are in the same direction, $\mathcal{C}(\alpha)$ is empty with high probability for sufficiently large $m$ when the postulated model is misspecified. More irregular departures from uniformity may be better detected by alternative combination rules \citep[see e.g.][]{Birnbaum1954, Heard2018}.

\subsection{Two parallel analyses}\label{secTwo}

We analyse the behaviour of the confidence set \eqref{eqConfSet}, assuming the existence of random variables $U_j(\psi_0)$. The aim is to supply qualitative insight through a simple but reasonably general analysis, noting that the route to constructing $U_j(\psi_0)$ is necessarily problem-specific. Calculations at this level of generality are inevitably idealised.

When the model is erroneous or if it contains the true distribution but $\psi_0\neq \psi^*$, then $U_j(\psi_0)$, $j=1,\ldots,m$, are by definition not standard uniform, and for most of the example cases we have in mind, are not identically distributed either. Since $U_j(\psi_0)$ is necessarily supported on $[0,1]$, insight can be obtained by approximating its density function using the parametric family 
\begin{equation}\label{eqPara}
(1-\vartheta_j)u^{-\vartheta_j}, \quad  0<u<1, \quad  0< \vartheta_j<1, 
\end{equation}
so that the null distribution is recovered in the limit as $\vartheta_j\rightarrow 0$. The model \eqref{eqPara} is deliberately oversimplified in that it specifies the directions of departure from uniformity to be in the same direction for all $j$. Specifically, the distribution \eqref{eqPara} concentrates towards zero for all $\vartheta_j$, the opposite behaviour represented only when $U_j(\psi_0)$ is replaced by $1-U_j(\psi_0)$.

Under \eqref{eqPara}, the $k$th moment of $R_j:=-\log U_j(\psi_0)$ has the convenient form
\begin{align*}
\EE(R_j^k) =  \frac{\Gamma(k+1)}{(1-\vartheta_j)^{k}}, \quad k\in\{1,2,\ldots\}.
\end{align*}
The mean and variance of the aggregate $A_m:=2\sum_{j=1}^m R_j$ are
\begin{equation}\label{eqMoments}
\EE(A_m)=2\sum_{j=1}^{m}(1-\vartheta_j)^{-1}, \quad \VV(A_m) 
= 4 \sum_{j=1}^{m}(1-\vartheta_j)^{-2},
\end{equation}
recovering the $\chi^2_{2m}$ mean and variance as $\vartheta_j\rightarrow 0$ for all $j$ and showing that $A_m$ has a larger mean and variance than a $\chi^2_{2m}$ random variable, with the discrepancies increasing as $m\rightarrow \infty$. 

On letting $k_\alpha$ be the $1-\alpha$ quantile of the $\chi^2_{2m}$ distribution, Markov's inequality in the form presented in Appendix \ref{secMarkov} shows that
\begin{equation}\label{eqMarkov}
\text{pr}(A_m \geq k_{\alpha}) \geq 1 -  \inf_{s>0} \biggl(e^{s k_{\alpha}}\prod_{j=1}^m \frac{1-\vartheta_j}{1-\vartheta_j + 2s}\biggr).
\end{equation}
The constituent terms are in the interval $(0,1)$ for all $s$ in the permissible range, so the product tends to zero exponentially fast in $m\rightarrow \infty$.

For a parallel analysis that does not involve specifying a parametric family for $U_{j}(\psi_0)$, let $\mu_m(\psi_0) = \sum_{j=1}^m\EE(R_j)$ and $\tau_m(\psi_0)=(\sum_{j=1}^m\text{var}(R_j))^{1/2}$. Then by the central limit theorem $\tau_m(\psi_0)^{-1}(\sum_{j=1}^m R_j - \mu_m(\psi_0))$ is asymptotically standard normal for large $m$, thus
\begin{equation}\label{eqStandardisation}
\text{pr}(A_m \geq k_{\alpha}) \simeq 1- \Phi\biggl(\frac{k_\alpha - 2\mu_m(\psi_0)}{2\tau_{m}(\psi_0)}\biggr) \simeq 1- \Phi\biggl(\frac{2m + 2 \sqrt{m}z_\alpha - 2\mu_m(\psi_0)}{2\tau_{m}(\psi_0)}\biggr),
\end{equation}
as $m\rightarrow \infty$, where $z_\alpha$ is the $1-\alpha$ quantile of the standard normal distribution and where we have used the normal approximation to the $\chi^2_{2m}$ distribution in the form $k_\alpha\simeq 2m + 2\sqrt{m}z_\alpha$ for large $m$. Under the null, where the model is correct and $\psi_0=\psi^*$, $U_j(\psi_0)$ has a standard uniform distribution for all $j$, and $R_j$ has a unit exponential distribution, implying that $\mu_m(\psi_0)=m$ and $\tau_m(\psi_0)=\sqrt{m}$. In this case, the previous display reduces to $\text{pr}(A_m\geq k_\alpha)\simeq \alpha$, which is the exact answer for any $m$. When the null distribution is violated, $\text{pr}(A_m \geq k_\alpha)\rightarrow 1$ if and only if $(k_\alpha-2\mu_m(\psi_0))/2\tau_m(\psi_0)\rightarrow -\infty$, which requires that the mean grows faster than the standard deviation as $m\rightarrow \infty$. There is a parallel calculation for the left tail.

\subsection{Model assessment using a $\mathcal{M}_0$-consistent estimator of $\psi^*$}\label{secEstimateParam}

When a consistent estimator of $\psi$ is available under $\mathcal{M}_0$, and when the effective sample size is large enough that sampling variability of $\hat\psi$ around $\psi^*$ does not materially affect the null distribution, there is no need to assess $U_j(\psi_0)$ for all plausible values $\psi_0$; instead $\hat\psi$ can be used in place of $\psi_0$, leading to the statistic
\begin{equation}\label{eqHatVersion}
\hat{A}_m=-2\sum_{j=1}^m \log U_j(\hat\psi),
\end{equation}
and its complementary version with $1-U_j$ in place of $U_j$. If $\hat A_m$ is extreme when calibrated against the $\chi^2_{2m}$ distribution, that casts doubt on the adequacy of the model. The discussion following \eqref{eqPara} illustrates that the versions based on $1-U_j$ and $U_j$ need not be interchangeable for the purpose of refuting a false postulated model, one version tending to have higher power than the other depending on the form of departure. 

Under the postulated model with true parameter value $\psi^*$, let $f$ be a function such that $U_j(\psi^*)=f(X_j, \psi^*)$ for $j=1,\ldots,m$ has a standard uniform distribution. The continuous random variable $X_j$ may itself be an induced replicate according to one of the definitions in \S \ref{secInduced}, in which case the only role of $f$ is to transform the distribution to standard uniformity. If $X_j$ is not an induced replicate, for instance if conditioning on a sufficient statistic for a nuisance parameter is required, then $f$ is part of $h_0$ in Definition \ref{defWeakReplInducement}. 

Let $\hat \psi$ be a null-consistent estimator of $\psi^*$ as the effective information $n\rightarrow \infty$, where $n$ may be, but need not be, equal to $m$. For simplicity of exposition let $\psi$ be a scalar, the main difficulties in extending to vector-valued $\psi$ being notational. The usual linearisation argument gives, for each $j=1,\ldots, m$
\begin{equation}\label{eqUhat}
U_j(\hat\psi) = U_j(\psi^*) + o_p(1), \quad (n\rightarrow \infty)
\end{equation}
provided that the $\psi$-derivative of $f$ is almost-surely bounded at $X_j$ in a neighbourhood of $\psi^*$. In other words, $U_j(\hat\psi) \leadsto U_j(\psi^*)$ where $\leadsto$ means weak convergence. If $n \asymp m$ as in Example \ref{exIntro}, the above linearisation is ineffectual for showing that $\hat A_m$ has an approximate $\chi^2_{2m}$ distribution. In the notation of equation \eqref{eqStandardisation}, $\mu_m(\psi^*)=m$ and $\tau_m(\psi^*)=\sqrt{m}$, thus the argument presented in equation (1.3) of \citet{Randles1982} and the approximation $k_\alpha\simeq 2m + 2\sqrt{m}z_\alpha$ used in \eqref{eqStandardisation} implies
\begin{equation}\label{eqNormal}
\pr (\hat A_m > k_\alpha) \simeq 1 - \Phi\biggl(\frac{k_\alpha - 2m}{2\sqrt{m}}\biggr) \simeq  \alpha, \quad (m\rightarrow \infty)
\end{equation}
provided that $\hat\psi \rightarrow_p \psi^*$ as $n(m)\rightarrow \infty$.

When the model is violated, $\hat\psi$ either converges in probability to a a fixed value $\psi_0^*\neq \psi^*$,  violating standard uniformity of each $U_j(\psi_0^*)$ and therefore invalidating the mean and variance standardisation used in \eqref{eqNormal}, or $U_j(\hat\psi)$ has a non-degenerate limit distribution. This second scenario not only violates asymptotic uniformity of $U_j(\hat\psi)$ for every $j$ but also induces a complicated dependence between the $U_j(\hat\psi)$, $j=1,\ldots,m$, because of the shared dependence on $\hat\psi$, whose distribution typically depends on nuisance parameters when the postulated model is erroneous. In general, any such dependence is expected to exacerbate the distributional violations in the aggregate \eqref{eqHatVersion} or \eqref{eqNormal}, so may well be beneficial for detecting departures from the model.

\section{Approximate replication in semiparametric regression}\label{secExpFamilyReg}

We now turn to model \eqref{eqGeneralForm} of \S \ref{secIntro}. A typical semiparametric analysis would make smoothness or other assumptions on $h$, the most restrictive being local constancy of $h$. This produces some internal replication in the nuisance component within blocks of covariate observations, and an inducible replication between blocks, as all individuals $i\in \mathcal{I}_j$ for whom $x_i \in \mathcal{X}_j$, say, have an outcome model
\begin{equation}\label{eqLocalConstant}
g(\EE(Y_i \mid x_i, t_i)) = \gamma_{j} + \beta t_i, \quad t_i\in\{0,1\}, \;\; \gamma_j = h(x_i), \;\; i \in \mathcal{I}_j.
\end{equation}

Since local constancy is an unnatural assumption, we seek to understand the error induced in the relevant conditional analysis through making this simplifying assumption. To this end, suppose, notionally, that the local constancy assumption is satisfied. The jointly sufficient statistics for $\beta$ and $(\gamma_j)_{j=1}^m$ are $V:=\sum_{i=1}^n Y_i t_i$ and $(S_j)_{j=1}^m$, where $S_j:=\sum_{i\in \mathcal{I}_j} Y_i$ and the theory of exponential families \citep[e.g.][p.~130]{BN1978} ensures that conditioning eliminates $(\gamma_j)_{j=1}^m$, the conditional density function being of the form
\begin{equation}\label{eqCondExp}
f_{V|S_1,\ldots, S_m}(v\mid s_1^o,\ldots, s_m^o;\beta)= \exp\{v\beta - \overline{\kappa}(\beta;s_1^o,\ldots, s_m^o)\}f_0(v;s_1^o,\ldots, s_m^o),
\end{equation}
where, $f_0$ is a ratio of the corresponding terms in the exponential family forms for $(V,S_1,\ldots,S_m)$ and $(S_1,\ldots, S_m)$ respectively, and, with $\mu$ the appropriate base measure,
\[
\overline{\kappa}(\beta;s_1^o,\ldots, s_m^o)=\log \int \exp(v\beta)f_0(v;s_1^o,\ldots, s_m^o) d\mu(v),
\]
which is a log summation if $v$ is discrete.

Two aspects are of interest: the extent to which the true conditional density function $f^*_{V|S_1,\ldots, S_m}(v\mid s_1^o,\ldots, s_m^o;\beta, h)$ is approximable by \eqref{eqLocalConstant}, this being the basis for conditional likelihood estimation of $\beta$; and the extent to which the conditional distribution function $F^*_{V_j|S_j}(v\mid s_j^o;\beta, h)$ is approximable by $F_{V_j|S_j}(v \mid s_j^o;\beta)$, where $V_j:=\sum_{i\in\mathcal{I}_j} Y_i t_i$.

\begin{proposition}\label{propApproxError2}
	Suppose that $(Y_i)_{i=1}^n$ are generated by a distribution in the exponential-family regression model \eqref{eqGeneralForm}, where $h$ is H{\"o}lder continuous. Then for each $x_i \in \mathcal{X}_j$, $h(x_i)=\gamma_j+\delta_i$. On approximating the conditional expected respond by \eqref{eqLocalConstant}, the conditional density function $f^*_{V|S_1,\ldots, S_m}(v\mid s_1^o,\ldots, s_m^o;\beta, h)$ can be approximated by \eqref{eqCondExp}, with approximation error quantified by
	\begin{equation}\label{eqApproxError}
	\log\EE\biggl[\frac{f^*_{V|S_1,\ldots, S_m}(v(Y_1,\ldots, Y_n)\mid s_1^o,\ldots, s_m^o;\beta, h)}{f_{V|S_1,\ldots, S_m}(v(Y_1,\ldots, Y_n)\mid s_1^o,\ldots, s_m^o;\beta)}\biggr]= \sum_{j=1}^m \sum_{i\in \mathcal{I}_j}\sum_{k=1}^\infty\frac{\delta_{i}^{k}}{k!}(\kappa_{i,k}^* - \kappa_{i,k}),
	\end{equation}
	where $\kappa_{i,k}^*=\kappa^{(k)}(t_i\beta + h(x_i))$, the $k$th cumulant under the true distribution in model \eqref{eqGeneralForm} and $\kappa_{i,k}=\kappa^{(k)}(t_i\beta + \gamma_j)$, the $k$th cumulant under the approximating model \eqref{eqLocalConstant}. For the conditional density functions within the $j$th partition,
	\begin{equation}\label{eqApproxErrorPartition}
	\log\EE\biggl[\frac{f^*_{V_j|S_j}(V_j\mid s_j^o;\beta, h)}{f_{V_j|S_j}(V_j\mid s_j^o;\beta)}\biggr]=  \sum_{i\in \mathcal{I}_j}\sum_{k=1}^\infty\frac{\delta_{i}^{k}}{k!}(\kappa_{i,k}^* - \kappa_{i,k}).
	\end{equation}
\end{proposition} 	 

\begin{remark}
	The log-expectation of the multiplicative error in \eqref{eqApproxError} provides the cleanest statement in terms of cumulant differences. Lemma \ref{lemmaApproxError} in the supplement, combined with \eqref{eqCumExp} below, also provide a convergence rate when the log and expectation operator are reversed; see Proposition \ref{propReverse}.
\end{remark}

Since in \eqref{eqApproxError} and \eqref{eqApproxErrorPartition}, 
\begin{equation}\label{eqCumExp}
\kappa^*_{i,k} = \kappa_{i,k} + \sum_{t=1}^\infty \frac{\delta_i^t \kappa_{i,k+t}}{t!},
\end{equation}
the $k=1$ leading-order term in the $i$th summand of \eqref{eqApproxError} and \eqref{eqApproxErrorPartition} is of order $\delta_{i}^2$. The following corollary is immediate from well known results in approximation theory. In particular, if the observation range for $x$ is rescaled to $[0,1]$, then the error in approximating $h$ by a piecewise constant function $h_\text{p}\in \mathcal{H}(\Delta_m)$ in the equally-spaced $m$-partition $\Delta_m$ is \citep[][p.~62]{DeVore1998} $\inf_{h_\text{p} \in \mathcal{H}(\Delta_m)}\|h-h_\text{p}\|_{\infty}= O(m^{-\alpha})$, where $\alpha\in(0,1]$ is the H{\"o}lder parameter with $\alpha=1$ corresponding to Lipschitz continuity. 

\begin{corollary}\label{corollApproxError}
	The approximation error \eqref{eqApproxError} is, without further assumptions
	\begin{equation}\label{eqFull}
	\log\EE\biggl[\frac{f^*_{V|S_1,\ldots, S_m}(v(Y_1,\ldots, Y_n)\mid s_1^o,\ldots, s_m^o;\beta, h)}{f_{V|S_1,\ldots, S_m}(v(Y_1,\ldots, Y_n)\mid s_1^o,\ldots, s_m^o;\beta)}\biggr]=  O\bigl(m^{-2\alpha} n\bigr) 
	\end{equation}
	in the full conditional analysis, and
	\begin{equation}\label{eqPart}
	\log\EE\biggl[\frac{f^*_{V_j|S_j}(V_j\mid s_j^o;\beta, h)}{f_{V_j|S_j}(V_j\mid s_j^o;\beta)}\biggr]=  O\bigl(m^{-3\alpha} n\bigr)
	\end{equation}
	in the conditional analysis within the $j$th partition.
\end{corollary}

Although Proposition \ref{propApproxError2} provides insight via the cumulant expansion, the correct metric for likelihood analysis, exchanges the logarithm and expectation operator.

\begin{proposition}\label{propReverse}
	In the notation of Proposition \ref{propApproxError2}, let $\ell_m(\beta) = m^{-1}\sum_{j=1}^m \log f_{V_j|S_j}(V_j\mid s_j^o;\beta)$ be the scaled conditional log-likelihood function for $\beta$ under the piecewise constant approximation, and let $\ell_m^*(\beta;h) = m^{-1}\sum_{j=1}^m \log f^*_{V_j|S_j}(V_j\mid s_j^o;\beta; h)$ be the inaccessible scaled conditional log-likelihood function treating $h$ as known. Then the approximation error in the objective function over compact $\mathcal{B}$ is
	\begin{equation}\label{eqLikApprox}
	\sup_{\beta \in \mathcal{B}}|\EE \{\ell_m(\beta)\} - \EE \{\ell_m^*(\beta; h)\}|= O\bigl(m^{-3\alpha} n\bigr).
	\end{equation}
\end{proposition}

It follows from Proposition \ref{propReverse} that $m\gtrsim n^{1/3\alpha}$ is required to ensure that \eqref{eqLikApprox} decreases with sample size, implying that a H{\"o}lder exponent $\alpha\leq \tfrac{1}{3}$ is not permitted. A trade off here is between approximation error and risk of degeneracy from conditioning. The usual trade off between estimation error and approximation error emerges once one seeks estimation of $\beta$ based on $\ell_m(\beta)$. \citeauthor{Anderson1970}'s (\citeyear{Anderson1970}) proof of consistency and asymptotic normality of the conditional maximum likelihood estimator implies a rate of $m^{-1/2}$ in a uniform law of large numbers for $\ell_m(\beta)$. Equalising approximation error and estimation error then suggests taking $m\asymp n^{2/(6\alpha-1)}$. For Lipschitz continuous $h$, this choice gives $m\asymp n^{2/5}$.

Let
\[
U_j^*(\beta):= F^*_{V_j|S_j}(V_j\mid s_j^o;\beta, h), \quad U_j(\beta):=F_{V_j|S_j}(V_j\mid s_j^o;\beta).
\]
Under the true distribution in the postulated model \eqref{eqGeneralForm}, $U_j^*(\beta^*)$, $j=1,\ldots, m$, are standard uniform replicates when $V_j$ is continuous conditional on $S_j=s_j^o$, and $U_j(\beta_0)$  are postulate replicates under the approximating piecewise-constant model \eqref{eqLocalConstant}. For discrete random variables, a continuity correction restores uniformity \citep[e.g.][]{Czado2009}. Proposition \ref{propReplicates} shows how the error from the piecewise-constant approximation propagates to the scale of the aggregated replicates
\[ 
A_m^*(\beta)=-2\sum_{j=1}^m\log U_j^*(\beta), \quad \;\; A_m(\beta)=-2\sum_{j=1}^m\log U_j(\beta),
\]
where $A_m$ is \citeauthor{Fisher1932}'s (\citeyear{Fisher1932}) statistic as defined in \S \ref{secFramework}.

\begin{proposition}\label{propReplicates}
	The mean and variance of the aggregated postulate replicates satisfy
	\begin{align*}
	\EE \{A_m(\beta)\} - \EE\{A_m^*(\beta)\} = &\;  O(m^{-2\alpha} n) \\
	\VV \{A_m(\beta)\} - \VV\{A_m^*(\beta)\} = &\;  O(m^{-2\alpha} n).
	\end{align*}
\end{proposition}

The statistic $A_m^*(\beta^*)$ has a $\chi^2_{2m}$ when model \eqref{eqGeneralForm} is adequate, and Proposition \ref{propReplicates} quantifies the departure from this in the first and second moments when the piecewise-constant approximation \eqref{eqLocalConstant} is used.

\section{Further insights from idealised cases}\label{secMP}

\subsection{Inferential structures for matched-pair examples}

Section \ref{secExpFamilyReg} illustrated the considerations involved in a typical semiparametric analysis, where an approximation was needed for inducible replication. In this section we consider some idealised settings whether the block structure of the design implies an inducible replication. 

The simplest such setting is the matched-pair design as in Example \ref{exIntro} leading to outcomes $(Y_{j1}, Y_{j0})_{j=1}^m$ on the treated and untreated units for the $m$ pairs. These will serve as useful cases for some perturbative power analyses in \S \ref{secPowerPairs} but their primary purpose in this section is to illustrate additional semiparametric structures for replication inducement.

\begin{example}[Continuation of Example \ref{exIntro}]\label{exIntroCont}
	The statistic $U_j(\psi_0)$ constructed in Example \ref{exIntro} arises also through a complementary perspective, valuable in other settings. On replacing the observed value $s_j(\psi_0)$ by the corresponding random variable, we obtain
	\begin{equation}\label{eqEquiv}
	U_j(\psi_0) = \frac{Y_{j1}\psi_0}{S_{j}(\psi_0)} = \frac{\psi_0 Z_j}{1+ \psi_0 Z_j},
	\end{equation}
	where $Z_j=Y_{j1}/Y_{j0}$. The distribution function of $Z_j$ under the proportional rates model of Example \ref{exIntro} is
	\begin{equation}\label{eqDistExpMult}
	F_{Z}(z;\psi^*) = \frac{\psi^{*} z}{1+ \psi^{*} z}.
	\end{equation}
	Thus, \eqref{eqEquiv} is $F_{Z}(Z_j;\psi_0)$, showing through a different route that uniformity is achieved if and only if the postulated model holds with $\psi_0=\psi^*$, or the true distribution over $(Y_{j1}, Y_{j0})$ is different from that postulated but induces the same distribution \eqref{eqDistExpMult} over $Z_j$. The transformation to $Z_j=Y_{j1}/Y_{j0}$ is the most immediate way to eliminate the nuisance parameters in this scale model. It was not a priori obvious, however, that this simple route to internal replication preserves the co-sufficient information, in the generalised sense used in Example \ref{exIntro}. The set of random variables $Z_1,\ldots,Z_m$ might be viewed from a different perspective as ancillary for the nuisance parameters, in the sense that, from observation of these statistics alone, estimation of $\gamma_1,\ldots,\gamma_m$ is impossible. \hfill$\square$ 
\end{example}

\begin{example}\label{exampleAdditiveRate}
	Suppose that the true distribution of $(Y_{j1}, Y_{j0})$ belongs to an additive exponential model with rates $\gamma_j+\Delta$ for the treated individual in pair $j$ and $\gamma_j$ for the untreated individual, with true value $\Delta^*$ for the treatment parameter. This model would arise, for instance, for the first event time in a Poisson process, where an additional source of events is present in the treatment group.  In the additive model, the statistic $S_j =Y_{j1}+Y_{j0}$ is sufficient for the nuisance parameter $\gamma_j$ and the conditional distribution function of $Y_{j1}$ given $S_{j}=s_j$ is 
	\begin{equation}\label{conditionalDist}
	F_{Y_{j1} | S_j }(y \cond s_j \, ; \Delta^*) = \frac{1- e^{-\Delta^* y}}{1-e^{-\Delta^* s_j}}.
	\end{equation}
	Provided that the model is correctly specified, each $U_{j}(\Delta^*):=F_{Y_{j1} | S_j}(Y_{j1} \cond s_j \, ; \Delta^*)$ is standard uniformly distributed conditional on $S_j=s_j$, and therefore also unconditionally. The above conclusion is invalidated if $\Delta^*$ is replaced by another value $\Delta_0$. \hfill$\square$
\end{example}

Replication-inducing constructions for matched pairs fall into two broad classes. The first eliminates pair-specific nuisance parameters by conditioning on a sufficient statistic, as illustrated in Examples \ref{exIntro} and \ref{exampleAdditiveRate}; the second eliminates them by transformation to a random variable whose distribution is free of nuisance parameters, as illustrated by Example \ref{exIntroCont}. Each of these may or may not depend on the parameter of interest, leading to four generic inferential structures. Ancillarity for a nuisance parameter is not a standard concept, and its relevance for model assessment is an observation new to this work.

\begin{structure}\label{structureSuff1}
	The joint density function $f(y_{1}, y_0; \gamma_j, \psi)$ of $(Y_{j1}, Y_{j0})$ is such that there is a statistic $S_j$, not depending on $\psi$, such that $S_j:=s(Y_{j1},Y_{j0})$ is sufficient for $\gamma_j$ and the conditional distribution function of $F_{Y_{j1}|S_j}(y \cond s_j; \psi)$ is continuous and bijective in $y$ for any $\psi$, and injective in $\psi$ for any fixed $y$.
\end{structure}

\begin{remark}
	It is immaterial whether $Y_{j1}$ or $Y_{j0}$ is used in the conditional distribution in Structure \ref{structureSuff1}. As seen from Figure \ref{figCosuff}, uniformly distributed variation in $Y_{j1}$ on the co-sufficient manifold implies uniform variation in $Y_{j0}$.
\end{remark}

\begin{structure}\label{structureSuff2}
	The second co-sufficiency structure parallels Structure \ref{structureSuff1}, except that the sufficient statistic depends on the interest parameter $\psi$, as in Example \ref{exIntro}.
\end{structure}

\begin{structure}\label{structureTrasform1}
	The joint density function $f(y_{1}, y_0; \gamma_j, \psi)$ of $(Y_{j1}, Y_{j0})$ is such that there is a transformation $t$, not depending on $\psi$, such that $Z_j:=t(Y_{j1},Y_{j0})$ has a distribution function $F_Z(z;\psi)$ that is continuous and bijective in $z$ for any fixed $\psi$, injective in $\psi$ for any fixed $z$, and functionally independent of the nuisance parameter $\gamma_j$. While Structure \ref{structureTrasform1} only applies exactly in the case of transformation models, Structure \ref{structureTrasform2} extends beyond these classes, as illustrated below in Example \ref{exampleGamma}.
\end{structure}

\begin{structure}\label{structureTrasform2}
	Structure \ref{structureTrasform2} parallels Structure \ref{structureTrasform1}, except that the transformation depends on the interest parameter $\psi$.
\end{structure}

The following proposition formalises the sense in which each structure induces internal replication. The proposition refers, for each of Structures \ref{structureSuff1}--\ref{structureTrasform2}, to equivalence classes of density functions within which the approach has no power to reject the postulated model. If the true density function violates the postulated model but belongs to the stated equivalence class, then the model violation will not be detected at any sample size. An analogy is to Markov equivalence classes in Gaussian graphical models, where multiple causal explanations give rise to the same distribution over the outcomes.

\begin{proposition}\label{propPairs}
	Let $(Y_{j1}, Y_{j0})$ have a joint density function, provisionally assumed to belong to the pairwise model $\mathcal{M}_j=\{f(y_{j1}, y_{j0}; \gamma_j, \psi): \gamma_j \in \Gamma, \psi \in \Psi\}$, where $f=f_j$ satisfies one of Structures \ref{structureSuff1}--\ref{structureTrasform2}. When any such structure holds, define, respectively
	\begin{align*}
	U^{(1)}_j(\psi_0) &:= F_{Y_{j1} \mid S}(Y_{j1} \,\big|\, s_j;\psi_0),
	& & S_j = s(Y_{j1}, Y_{j0}), \\
	U^{(2)}_j(\psi_0) &:= F_{Y_{j1} \mid S_j(\psi_0)}(Y_{j1} \,\big|\, s_j(\psi_0)),
	& & S_j(\psi_0) = s(Y_{j1}, Y_{j0}; \psi_0), \\
	U^{(3)}_j(\psi_0) &:= F_{Z_j}(Z_j;\psi_0),
	& & Z_j = t(Y_{j1}, Y_{j0}), \\
	U^{(4)}_j(\psi_0) &:= F_{Z_{j}(\psi_0)}(Z_j(\psi_0)),
	& & Z_j(\psi_0) = t(Y_{j1}, Y_{j0}; \psi_0).
	\end{align*}
	Then $U_j^{(k)}(\psi_0)\sim U(0,1)$ if and only if $\mathcal{M}_j$ holds with true parameter value $\psi^*=\psi_0$, or if the true distribution of $(Y_{j0}, Y_{j1})$ with density function $f^*_j$ belongs to the equivalence class $\mathcal{E}^{(k)}(\psi_0)$, as specified in Appendix \ref{secEquiv}. 
\end{proposition}

\begin{proof}
	A proof is given in the appendix.
\end{proof}

\begin{corollary}
	Suppose that the postulated semiparametric model has pair-specific distributions described by one of Structures \ref{structureSuff1}--\ref{structureTrasform2}. There is information to detect departures from the postulated model at $\psi_0$ unless the true pair-specific density $f^*_j$ belongs to the corresponding equivalence class $\mathcal{E}(\psi_0)$ for all $j=1,\ldots,m$.
\end{corollary}

In the context of Example \ref{exIntroCont}, the relevant equivalence class is $\mathcal{E}^{(3)}(\psi_0)$ and consists of all distributions over the original random variables $(Y_{j1},Y_{j0})$ that induce the same distribution over the ratios $Z_j=Y_{j1}/Y_{j0}$ as that postulated under $\mathcal{M}_j$, namely the distribution described by $F_Z(z; \psi_0)$ in \eqref{eqDistExpMult}. Since $Z_1,\ldots, Z_m$ are identically distributed under the postulated model, members of the equivalence class, if any such member exists, must be such that any pair-specific heterogeneity is eliminated through the same operation as under the postulated model, as otherwise no replication would be induced and standard uniformity for every pair could not be achieved. Elimination of pair dependence under the same operation as under the postulated model is a necessary but not a sufficient condition for a density function $f_j$ to belong to the relevant equivalence class. Consider, for instance, a true distribution belonging to the Weibull family with multiplicative treatment effect on the rate scale and a postulated model that is exponential, a special case. Ratios $Z_j=Y_{j1}/Y_{j0}$ eliminate the pair parameters in both cases, but no members of the Weibull model with non-unit shape parameter belong to the equivalence class $\mathcal{E}^{(3)}(\psi_0)$ for any value of $\psi_0$, as the induced distribution over $Z_j$ is different.

While Proposition \ref{propPairs} ensures appropriate coverage when the postulated model is correct, it does not automatically translate to a powerful method for rejecting false models, since when the postulated model is violated at $\psi_0$, the notional replicates $U_j^{(k)}(\psi_0)$ typically depend on the nuisance parameters $\gamma_j$, and the summation in \eqref{eqConfSet} does not necessarily produce an aggregate that is sufficiently far in the tail of a $\chi^2_{2m}$ distribution. It is partly for this reason that the version based on $\hat\psi$ from \S \ref{secEstimateParam} is preferred.

\subsection{Local power analyses}\label{secPowerPairs}

To probe sensitivity, we consider cases in which the true distribution belongs to a model that departs only slightly from that postulated. In Example \ref{exampleTrueWeibull}, the same preliminary reductions eliminate the nuisance parameters under both models, so that there is no excess variability under the true model relative to that postulated, as discussed in the paragraph following \eqref{eqNormal}. The example illustrates the merits of using an estimate of $\psi$. 

\begin{example}\label{exampleTrueWeibull}
	Let the postulated model be the proportional rates model  of Example \ref{exIntro} and suppose that true distribution of the paired outcomes is Weibull with rate parameters $\gamma_j$, and $\gamma_j\psi^*$, respectively for untreated and treated individuals in the $j$th pair. Suppose further that the shape parameter of these Weibull distributions is $\varsigma\neq 1$, so that the postulated model is misspecified. The same pairwise operation eliminates $\gamma_j$ under both models, the distribution and density functions of the ratios $Z_j=Y_{j1}/Y_{j0}$ being
	\begin{equation}\label{eqWeibullRatioDist}
	F_{Z}(z) = \frac{\psi^* z^\varsigma}{1+ \psi^* z^{\varsigma}},\ \quad f_{Z}(z) = \frac{\psi^*\varsigma z^{\varsigma-1}}{(1+ \psi^* z^{\varsigma})^2}.
	\end{equation}
	By monotonicity, the distribution function of $U_j(\psi_0)$ in \eqref{eqEquiv} is
	\begin{align}\label{eqMono}
	\text{pr}(U_j(\psi_0) \leq u)= \text{pr}\Bigl(Z_j \leq \frac{u}{\psi_0(1-u)} \Bigr) = \frac{\psi^* u^\varsigma}{\psi_0^{\varsigma}(1-u)^{\varsigma}+\psi^* u^{\varsigma}}, \qquad \;\;\; 0\leq u \leq 1,
	\end{align}
	which is non-uniform for all $\varsigma\neq1$ and all $\psi_0$.
	
	Although no choice of $\psi_0$ gives uniformity, the expectation of 
	$R_j=-\log U_j(\psi_0)$ can be shown via \eqref{eqIntParts} to be monotonically decreasing in $\psi_0$, implying the existence of a unique value of $\psi_0$ at which $\EE(R_j)=1$. Consequently, departures from exponentiality in the direction of the Weibull model are difficult to detect using the aggregation criterion \eqref{eqConfSet}. 
	
	Appendix \ref{secWeibullEstPsi} presents a power analysis when $\psi$ is estimated using the assumed density function of the ratios in \eqref{eqDistExpMult}. Since the model is misspecified, consistency of the maximum likelihood estimator is not immediate, but Appendix~\ref{secWeibullDeriv} shows that its probability limit is $\psi_0^*=\psi^*+O(\varsigma-1)$, from which a power analysis follows from \eqref{eqMono}. This shows that $\hat{A}_m$ in \eqref{eqHatVersion} has power to detect departures from the model even though the maximum likelihood estimator of $\psi^*$, based on the transformed random variables $Z_j$, is consistent modulo $O(\varsigma-1)$ under the misspecified model. \hfill$\square$
\end{example}

In Example \ref{exampleTrueAdditive} of Appendix \ref{secAdditivePower}, the outcome distribution is the same in the true and postulated models but the treatment effect represents different types of departure from baseline, implying a different inferential separation. The example illustrates similar difficulties in detecting departures from the postulated model in this challenging perturbative case, improvable by using $\hat\psi$ in place of $\psi_0$, as discussed in \S \ref{secEstimateParam}.

Examples \ref{exampleTrueWeibull} and \ref{exampleTrueAdditive} are ones in which the postulated model is so close to the true one that there is little harm in using it as a basis for inference. More substantial violations of modelling assumptions are probed by simulation in \S \ref{secSimMP} of the supplement.

\subsection{Unbalanced strata}\label{secTwoGroup}

Matched pairs are special cases of a more general two-group problem, frequently arising in observational settings. Specifically, treated and untreated individuals are stratified into groups as similar as possible. This typically leads to unbalanced strata, with different numbers of individuals in strata, and in treated and control groups per stratum. 

Inference is based on the sufficient statistics $S_{j1}$ and $S_{j0}$ within treatment groups and strata, which can be treated in a similar way to the paired observations of \S \ref{secMP}, as the strata sizes $r_{j1}$ and $r_{j0}$ are known. Let $(Y_{ij1})_{i=1}^{r_{j1}}$ and $(Y_{ij0})_{i=1}^{r_{j0}}$ be observations within the $j$th stratum for treated and untreated individuals respectively. As a first example, if the observations $(Y_{ij1})_{i=1}^{r_{j1}}$ and $(Y_{ij0})_{i=1}^{r_{j0}}$ are normally distributed with means $\gamma_{j}+\psi^*$ and $\gamma_{j}$ and variance $\tau$, the likelihood contribution to the $j$th stratum depends on the data only through $S_{j1}=\sum_{i=1}^{r_{j1}}Y_{ij1}/r_{j1}$ and $S_{j0}=\sum_{i=1}^{r_{j0}}Y_{ij0}/r_{j0}$. The difference $Z_j=\sum_{i=1}^{r_{j1}}Y_{ij1}/r_{j1}-\sum_{i=1}^{r_{j0}}Y_{ij0}/r_{j0}$ is normally distributed of mean $\psi^*$ and variance $\tau r_{j0} r_{j1}/(r_{j0}+r_{j1})$, which can be handled as discussed in \S \ref{secEstimateParam} with $\tau$ also estimated. This is an example of Structure \ref{structureTrasform1}, but the same answer is achieved via a conditioning argument based on Structure \ref{structureSuff1}.

If the individual observations are Poisson distributed counts with rates $\gamma_{j}\psi^*$ and $\gamma_{j}$, the sufficient statistics $S_{j1}$ and $S_{j0}$ are sums of these counts, Poisson distributed with rates $r_{j1}\gamma_{j}\psi^*$ and $r_{j0}\gamma_{j}$. The distribution of $S_{j1}$ or $S_{j0}$ conditional on $S_{j1}+S_{j0}$ eliminates $\gamma_j$ in the usual way. This is an example of Structure \ref{structureSuff1}. 

Example \ref{exampleGamma} is more interesting as it relies on Structure \ref{structureTrasform2} to eliminate the nuisance parameters.

\begin{example}\label{exampleGamma}
	If the originating variables are exponentially distributed with rates $\gamma_{j}\psi^*$ and $\gamma_{j}$, the sufficient statistics  $S_{j1}$ and $S_{j0}$ are gamma-distributed sums, with shape and rate parameters $(r_{j1}, \gamma_j\psi^*)$ and $(r_{j0},\gamma_j)$ respectively. Then
	\[
	Z_j(\psi^*):=\frac{r_{j0}\psi^* S_{j1}}{r_{j1}S_{j0}} 
	\]
	has an $F$ distribution with parameters $2r_{j1}$ and $2r_{j0}$. The strategy of \S \ref{secEstimateParam} based on Proposition \ref{propPairs} then applies directly. \qed
\end{example}

\section{Externally-induced replication}\label{secSynthetic}

Definition \ref{defExtReplInducement} is inspired by the external mechanism for replication, used by \citet{BRT}, as discussed in Example \ref{exampleBRT}. Cyrillic letters distinguish synthetic randomness from observable random variables. 

\begin{definition}[\textbf{Randomisation-induced  replication}]\label{defExtReplInducement}
	Let $\ya_1, \ldots, \ya_m$ be $m$ independent external sources of randomness. These are either random variables independent of $Y$, as in \citet{RasinesYoung2023}, or specify a component of randomness for independent replication from a suitable conditional distribution given $Y=y^o$, as in \citet{Ramdas2025}. If there exist random variables $(\tilde{V}_{j})_{j=1}^m$, constructed from $Y$ and  $\ya_j$ according to a rule $h_0$, that are iid under the postulated model, then $h_0$ is said to be \emph{replication-inducing through randomisation}.
\end{definition}

\begin{example}[Post-reduction inference for confidence sets of models]\label{exampleBRT}
\citet{BRT} sought to avoid sample splitting and post-reduction miscalibration of confidence sets by using the sufficiency/co-sufficiency separation presented in Example \ref{exampleRegression} in conjunction with an external source of randomness. They generated iid replicates $\tilde{Q}_1,\ldots, \tilde{Q}_k$ under the postulated model using a generalisation of the randomisation scheme introduced by \citet{RasinesYoung2023}. The details of that construction are not important to the present exposition. 

To reframe their approach using the unified exposition of \S \ref{secFramework}, consider the $m=k(k-2)/2$ inner products between all pairs $\tilde{Q}_i$, $\tilde{Q}_j$ of synthetic replicates. Under a hypothesised mean model with $d_0$ explanatory variables, these $m$ inner products (cosines of the angles between replicates), $Z_1,\ldots,Z_m$ say, each have density function \citep{Fisher1915}
\begin{equation}\label{eqFisher1915}
f_{Z}(z) =\frac{\Gamma((n-d_0)/2)}{\sqrt{\pi}\Gamma\{(n-d_0-1)/2\}} (1-z^{2})^{(n-d_0-3)/2}, \quad -1<z<1.
\end{equation}
The distribution function corresponding to \eqref{eqFisher1915} is
\begin{equation}\label{eqFisherDist1915}
F_{Z}(z) =\frac{1}{2}+\frac{2\Gamma((n-d_0)/2)\text{sign}(z)}{\sqrt{\pi}\Gamma\{(n-d_0-1)/2\}} \int_{0}^{z^2}u^{1/2-1}(1-u)^{(n-d_0-1)/2-1} du,
\end{equation}
where the definite integral is the incomplete beta integral with arguments $1/2$, $(n-d_0-1)/2$ and $z^2$. The $m$ inner products are converted to standard uniform random variables in the usual way $U_j = F_{Z}(Z_j)$, $j=1,\ldots m$. Although only the angles between disjoint pairs of replicates are independent, numerical work suggests that if $m$ is small relative to $n-d_0$, the dependence across all $m$ of them is sufficiently weak that the $\chi^2_{2m}$ null distribution of \eqref{eqHatVersion} holds to close approximation; see column 1 of Table \ref{tabCSM} in the supplementary file, where the empirical coverage probabilities of the model confidence sets and the number of false models in the set are reported.
\end{example}

The randomised inference of \citet{RasinesYoung2023} and \citet{BRT} hinged on Gaussianity of the outcome, extendable to multivariate ellipticity if the variance is known. Since for many parametric inference problems the derivative vector of the objective function is asymptotically normally distributed, there is, in principle, a route to approximate randomised inference based on the asymptotic normality of that limit distribution. We do not establish formal guarantees here, but the following example shows how randomised inference for the gradient vector would extend our earlier arguments to the proportional hazards model.

\begin{example}[Conditional sufficiency and proportional hazards]\label{examplePH} The proportional hazards model $\lambda(y;x)=\lambda_0(y)g(x_i^\T\beta)$ is unusual in that the infinite-dimensional nuisance parameter, the baseline hazard function $\lambda_0$, can be eliminated in the partial likelihood analysis for the regression coefficient $\beta$, circumventing its estimation \citep{Cox1972, Cox1975}. Here we frame partial likelihood in terms of a notion of conditional sufficiency for $\lambda_0$, pointing to a corresponding notion of conditional co-sufficiency for assessment of the model.

On each of a number of individuals there is a survival time $T_i$. If individuals are subject to non-informative censoring, the observable outcome consists of $Y_i=\min\{T_i,c_i\}$ and an indicator $d_i=1$ if $Y_i=T_i$, and $d_i=0$ otherwise. Our exposition here is for the uncensored setting $d_i=1$ for all $i$. 

Let $y_{(1)}<y_{(2)}<\cdots <y_{(n)}$ denote the ordered failure times and let $i_j$ be the index of the individual who fails at time $y_{(j)}$, so that $i=i_j$ if and only if $y_i=y_{(j)}$. Together, $y_{(1)},\ldots,y_{(n)}$ and $i_1,\ldots,i_n$ are equivalent to the unordered failure times $y_{1},\ldots,y_n$, in the sense that the latter can be recovered from the former, and vice versa. Let $\mathscr{R}_j := \{i: y_i\geq y_{(j)}\}$ be the set of individuals still at risk of failure at time $y_{(j)}$. The following arguments implicitly underpin the development of \citet{Cox1972}. 

By ignoring the indices $i_1,\ldots,i_n$, the ordered survival times $y_{(1)},\ldots,y_{(n)}$ are detached from any covariate dependence and therefore must, in the absence of the information supplied by $i_1,\ldots,i_n$, tell us nothing about how covariates influence the distribution of survival times, a version of ancillarity for $\beta$ in the presence of nuisance parameters. The distribution of these ordered times depends primarily on $\lambda_0$, the common instantaneous probability of failure at baseline.

Let $I_1,\ldots,I_{n}$ denote the random variables whose realisations are $i_1,\ldots,i_n$, and consider the conditional distribution of $I_1,\ldots,I_{n}$ given $y_{(1)},\ldots,y_{(n)}$. Write
\[
\mathscr{H}_j=\{y_{(1)},\ldots,y_{(j)},i_{1},\ldots,i_{j-1}\}
\]
for the history up to time $y_{(j)}$. The conditional probability that $I_j=i$, i.e.~the probability that a failure occurring at time $y_{(j)}$ is contributed by individual $i$, conditional on $\mathscr{H}_j$, is the familiar expression appearing in the partial likelihood for $\beta$: 
\begin{equation}\label{eqCoxPL}
p_{j,i}(\beta^*):=\text{pr}(I_j = i \mid \mathscr{H}_j) = \frac{\lambda(y_{(j)};x_i)}{\sum_{k\in\mathscr{R}_j} \lambda(y_{(j)};x_k)} = \frac{g(x_i^\T \beta^*)}{\sum_{k\in\mathscr{R}_j} g(x_k^\T \beta^*)}.
\end{equation}
Although the conditioning set is the entire history $\mathscr{H}_j$, \eqref{eqCoxPL} is free of $y_{(1)},\ldots, y_{(j-1)}$, a type of conditional sufficiency, for the baseline hazard function $\lambda_0$, of the increasing sets of order statistics $\{y_{(1)},\ldots, y_{(j-1)}\}$, conditional on a failure having occurred at time $y_{(j)}$.

The above discussion suggests an assessment of the proportional hazards assumption based on the conditional distributions of the indices $I_j$ given the relevant history $\mathscr{H}_j$. This is a version of co-sufficiency relative to the nuisance parameter $\lambda_0$, and leads to assessment of the model based on the same inferential partition that enabled inference on $\beta^*$, which is reasonable provided that the sample size is large enough to reliably estimate $\beta^*$. 

The usual partial-likelihood argument
\begin{align}
\begin{split}\label{eqPartialLik}
\text{pr}(I_1 = i_1, \ldots,I_{n}=i_n) =& \, \prod_{l=1}^n \text{pr}(I_l = i_l \mid i_{1},i_2,\ldots, i_{l-1}) \\
=& \, \prod_{l=1}^n \frac{g(x_{i_l}^\T \beta^*)}{\sum_{k\in\mathscr{R}_j} g(x_k^\T \beta^*)},
\end{split}
\end{align}
leads to a consistent estimator of $\beta^*$ with familiar properties \citep{Cox1972,WingWong}, but the discreteness of the indices and the changing risk sets present considerable difficulties for directly converting equations such as \eqref{eqCoxPL} into useable replicates for model assessment. The conditionally co-sufficient information is, however, contained in the partial likelihood, the partial likelihood-score being asymptotically normally distributed, suggesting randomised inference of the type proposed by \citet{RasinesYoung2023} and \citet{BRT}.
\end{example}

\section{Brief summary of numerical work} \label{secNumericalMain}

The supplementary file confirms the expected behaviour for several of the examples above, showing recovery of nominal error rates when the postulated model is true and high sensitivity to departures from an erroneous postulated model when a different semiparametric model generated the data. For brevity, we summarise one example that illustrates these properties. 

The postulated model is the inhomogeneous Poisson process discussed in Example \ref{examplePoisson} with semiparametric intensity function $\lambda_i(t)=e^{\gamma_i + \beta t}$ for the $i$th individual. Simulation is not straightforward and details are available in the supplement. To assess cases where the postulated model is erroneous, we also simulated from an inhomogeneous Poisson process with semiparametric intensity function specified differently as $\lambda_i(t)=e^{\gamma_i}t^\rho$ with $\rho\in\{-0.5, 0, 0.5, 1\}$.

By inducing an internal replication as discussed in Example \ref{examplePoisson}, the postulated model is assessed for its compatibility with the data and the results are reported in Table \ref{tabPoisson}, revealing high sensitivity even at values of $\rho$ very close to the point of intersection $\rho=0$. Left and right sensitivity refers to whether $U_j(\hat\psi)$ or $1-U_j(\hat\psi)$ was used in the construction of Fisher's statistic \eqref{eqHatVersion}.

{\footnotesize{

		\begin{table}[h!]
			
			\begin{center}
				\begin{tabular}{|c|c|cccccc|}
					\hline
					Direction     &	\multirow{2}{*}{$\rho$}     &  \multicolumn{6}{|c|}{Square root of number of individuals $n$} \\
					\cline{3-8}
					of sensitivity &				& $3$ & $4$ & $5$ & $6$ & $7$ & $8$  \\
					\hline		
					\textbf{left} & \textbf{0} & \textbf{0.03} & \textbf{0.04} & \textbf{0} & \textbf{0.03} & \textbf{0.02} & \textbf{0.05} \\
					\textbf{right} & \textbf{0}  & \textbf{0.03} & \textbf{0.03} & \textbf{0.02} & \textbf{0.03} & \textbf{0.01} & \textbf{0.03} \\
					\hline
					left & 0.1   & 0.18 & 0.34 & 0.42 & 0.61 & 0.73 & 0.92 \\
					right & 0.1 & 0.11 & 0.17 & 0.23 & 0.33 & 0.47 & 0.51 \\
					left & 0.2   & 0.72 & 0.88 & 0.98 & 1 & 1 & 1 \\
					right & 0.2 & 0.66 & 0.81 & 0.95 & 0.99 & 1 & 1 \\
					\hline
				\end{tabular}
				\medskip
				\caption{The generating process is the inhomogeneous Poisson process with intensity function $\lambda_i(t)=e^{\gamma_i + \rho \log t}$; the postulated model has intensity function $\lambda_i(t)=e^{\gamma_i + \beta t}$. Proportion of 200 simulation runs in which the 5\% test based on $F_{S_i|M_i}$ and \eqref{eqHatVersion} rejects the postulated model. The point of model intersection $\rho=0$ is highlighted.	\label{tabPoisson}}
				\vspace{-0.5cm}
			\end{center}
			
		\end{table}
}}

\section{Discussion and open problems}

The work is a new perspective on the assessment of semiparametric and highly parametrised models, illustrating the possibility and value, in settings where the model admits this, of circumventing estimation of the infinite- or high-dimensional component. The broad principle is to use the postulated model to induce replication of known form if and only if that model holds to an adequate approximation.

In each of Examples \ref{exIntro}--\ref{examplePH}, there is no replication in the original formulation, but there is inducible replication once structures in the models have been identified, and this does not correspond directly to sampling units. In the Poisson-process example, the individuals contribute a replicate observation after conditioning on a sufficient statistic. In semiparametric exponential-family regression, an approximation to the nonparametric component produces blocks of individuals who constitute a replicate observation after a similar type of conditioning argument, modulo approximation error. The matched-pair and two-group examples illustrate other important structures, namely, a parameter-dependent analogue and ancillarity for the nuisance components corresponding to the nonparametric part of the model. From these structures, replication can be induced from the paired observations. The proportional hazards model illustrates the role of conditional co-sufficiency and ancillarity, as well as externally-induced replication.

A reader has indicated another example that is amenable to our approach without inducing approximate error. In the semiparametric model for longitudinal data of \citet{DiggleFarewell}, replication can be induced through transformation to a cumulative scale. In the induced model, the nonparametric component is reduced to a martingale random effect, whose distribution is arbitrary and can be evaded in the analysis. 

The same reader has asked about the trade-off between avoiding estimation of a high-dimensional nuisance parameter, as favoured in the present paper, and having available an estimate of the nuisance parameter, even if the estimate is poor. Since the latter formulation entails the standard approach based on sample splitting into portions $a/n$ and $(n-a)/n$, it can be viewed as crude replication in which no convenient information partition is achieved. Each partition is informative about unknown parameters in proportion $\{a/n, (n-a)/n\}$ and informative about the model in proportion $\{a/n, (n-a)/n\}$. Inevitably, therefore, information about the model is sacrificed in the estimation step. By contrast, our partitions are such that all, or almost all, information about the model is preserved. 

Assessment of a statistical model for its compatibility with the data is inevitably a discrete problem unless competing models are nested, or can be artificially nested, allowing model space to be continuously traversed through variation of a parameter. By contrast, the simplest approaches to inference on the parameters of a given statistical model often involve maximisation of a log-likelihood or other objective function, and therefore typically invoke a notion of continuity on the parameter space. Perhaps for this reason, valuable inferential structures, such as interest-dependent co-sufficiency, and ancillarity for a nuisance parameter, have been overlooked in the literature, yet emerge naturally for model assessment.

\subsection*{Acknowledgements}

We are grateful to Daniel Farewell for comments on a first draft. HB was supported by a UK Engineering and Physical Sciences Research Council research fellowship (EP/T01864X/1). NR was supported by Natural Sciences and Engineering Research Council of Canada Grant RGPIN-2026-06428.

\bibliographystyle{plainnat}

\newpage

\begin{appendix}

	\section{Exact conditional distribution function in Example \ref{examplePoisson}}\label{secLaplace}

The conditional distribution function $F_{S_i|M_i}(s \mid m_i^o ; \beta)$ from Example \ref{examplePoisson} has support $s< m_i^o t_0$ and can be calculated exactly, although the expression is rather complicated.

The Laplace transform of \eqref{eqIndivContributins} is
\[
f^*_T(z)= \frac{\beta (e^{(\beta - z)t_0} - 1)}{(e^{\beta t_0}-1)(\beta - z)},
\]
thus the conditional density function of $S_i$ is, for any $\tau>\beta$, 
\begin{align*}
f_{S_i}(s \mid m_i ; \beta) = & \; \frac{\beta^{m_i}}{(e^{\beta t_0}-1)^{m_i}} \frac{1}{2\pi i}\int_{\tau - i\infty}^{\tau + i\infty} e^{z s}\frac{(e^{(\beta - z)t_0} - 1)^{m_i}}{(\beta - z)^{m_i}} dz \\
= & \; \frac{\beta^{m_i}}{(e^{\beta t_0}-1)^{m_i}} \sum_{v=0}^{m_i}{m_i\choose v} (-1)^{v} e^{v\beta t_0} 
\frac{1}{2\pi i}\int_{\tau - i\infty}^{\tau + i\infty} \frac{e^{sz} e^{-v z t_0} }{(z-\beta)^{m_{i}}} dz
\end{align*}
by the binomial formula. Let
\[
k^{*}(z)=\frac{e^{z(s-v t_0)} }{(z-\beta)^{m_{i}}}.
\]
Its contour integral is the residue at $z=\beta$ which, for a pole of order $m_i$, is 
\begin{align*}
\text{Res}(k^{*}(z), \beta) = &\; \frac{1}{(m_i - 1)!}\lim_{z\rightarrow \beta}\frac{d^{(m_{i}-1)}}{d z^{(m_{i}-1)}} \Bigl\{(z-\beta)^{m_i} k^*(z) \Bigr\} \\
= &\; \frac{1}{(m_i - 1)!}\lim_{z\rightarrow \beta}\frac{d^{(m_i -1)}}{d z^{(m_i-1)}}e^{z(s-v t_0)} \\
= &\; \frac{(s-v t_0)^{(m_i-1)}}{(m_i - 1)!}e^{\beta (s- vt_0)}.
\end{align*}
Thus,
\begin{align}
\nonumber f_{S_i}(s \mid m_i ; \beta) = & \;\frac{\beta^{m_i} e^{\beta s}}{(e^{\beta t_0}-1)^{m_i} \Gamma(m_i)} \sum_{v=0}^{m_i}{m_i\choose v} (-1)^{v}  (s-v t_0)^{(m_i-1)}\,\ind\{s> v t_0\},   \\
= & \;\frac{\beta^{m_i} e^{\beta s}}{(e^{\beta t_0}-1)^{m_i} \Gamma(m_i)} \sum_{v=0}^{\lfloor s/t_0\rfloor}{m_i\choose v} (-1)^{v}  (s-v t_0)^{(m_i-1)}, \quad s< m_i t_0
\end{align}
and zero otherwise. Integration gives
\[
F_{S_i|M_i}(s \mid m_i ; \beta) =  \frac{\beta^{m_i} }{(e^{\beta t_0}-1)^{m_i} \Gamma(m_i)} \sum_{v=0}^{\lfloor s/t_0\rfloor}{m_i\choose v}   (-1)^{v} e^{\beta v t_0} \int_0^{s-vt_0} e^{\beta w} w^{(m_i-1)}dw.
\]

\section{Derivation of equation \eqref{eqMarkov}}\label{secMarkov}

This is a standard argument. Let $k_\alpha$ be the $1-\alpha$ quantile of the $\chi^2_{2m}$ distribution, Markov's inequality implies
\[
\text{pr}(A_m \geq k_{\alpha}) = 1- \text{pr}(-A_m\geq -k_{\alpha}) \geq 1- e^{s k_{\alpha}}\EE(e^{-s A_m}), \quad s>0.
\]
Direct calculation shows that 
\[
\EE(e^{-s A_m}) = \prod_{j=1}^m \frac{1-\vartheta_j}{1-\vartheta_j + 2s}, \quad s> 0.
\]

\section{Proofs of Propositions \ref{propApproxError2} -- \ref{propReplicates}}

\begin{lemma}\label{lemmaApproxError}
	Suppose that $(Y_i)_{i=1}^n$ are generated by a distribution in the family \eqref{eqGeneralForm}, where $h$ is H{\"o}lder continuous. Then for each interval in the partition $\Delta_m$, $h(x_i)=\gamma_j+\delta_i$. On approximating by \eqref{eqLocalConstant}, the conditional density function $f^*_{V|S_1,\ldots, S_m}(v(y_1,\ldots, y_n)\mid s_1^o,\ldots, s_m^o;\beta,h)$ of $V$ at $v$ given $s_1^o,\ldots, s_m^o$, is approximable by \eqref{eqCondExp} with multiplicative error
	\begin{equation}\label{eqLemma}
	\frac{f^*_{V|S_1,\ldots, S_m}(v\mid s_1^o,\ldots, s_m^o;\beta,h)}{f_{V|S_1,\ldots, S_m}(v\mid s_1^o,\ldots, s_m^o;\beta)}=\exp\biggl\{\sum_{i=1}^n\delta_iy_i - \sum_{k=1}^\infty\sum_{j=1}^m \sum_{i\in \mathcal{I}_j}\frac{\delta_{i}^{k}\kappa^{(k)}(t_i\beta + \gamma_j)}{k!}\biggr\},
	\end{equation}
	where $\kappa(t_i\beta + \gamma_j)$ is the cumulant generating function of $Y_i$, $i\in \mathcal{I}_j$ under the postulated model piecewise constant model \eqref{eqLocalConstant}. Similarly, the conditional density function $f^*_{V_j|S_j}(V_j \mid s_j^o; \beta, h)$ is aprroximable by $f_{V_j|S_j}(V_j \mid s_j^o; \beta)$ with error
\begin{equation}\label{eqLemmaPart}
	\frac{f^*_{V_j|S_j}(v\mid s_j;\beta,h)}{f_{V_j|S_j}(v\mid s_j^o;\beta)}=\exp\biggl\{\sum_{i\in \mathcal{I}_j} \delta_i y_i - \sum_{k=1}^\infty \sum_{i\in \mathcal{I}_j}\frac{\delta_{i}^{k}\kappa^{(k)}(t_i\beta + \gamma_j)}{k!}\biggr\}.
	\end{equation}
\end{lemma}

\begin{proof}
	Write the true joint density function in canonical exponential family form as
	\[
	f^*_Y(y;\beta, h) = \exp\Bigl[\sum_{i=1}^n \Bigl\{y_i(t_i\beta + h(x_i)) - \kappa(t_i\beta+h(x_i))) \Bigr\} \Bigr]\prod_{i=1}^n f_0(y_i),
	\]
	the piecewise-constant approximation being
	\[
	f_Y(y;\beta, \gamma_1,\ldots,\gamma_m) = \exp\Bigl[\sum_{j=1}^m\sum_{i\in \mathcal{I}_j} \Bigl\{y_i(t_i\beta + \gamma_j) - \kappa(t_i\beta+\gamma_j)) \Bigr\} \Bigr]\prod_{i=1}^n f_0(y_i).
	\]
	The sufficient statistics under the piecewise constant model are $V:=\sum_{i=1}^n Y_i t_i$ and $S=(S_j)_{j=1}^m$, where $S_j:=\sum_{i\in \mathcal{I}_j} Y_i$. Let
	\[
	o(y)=\Bigl(\sum_{i=1}^n y_i t_i, \sum_{i\in \mathcal{I}_1}y_i, \ldots, \sum_{i\in \mathcal{I}_m}y_i\Bigr).
	\]
	Then with
	\[
	h_0(v,s) = \int_{\{y:o(y)=(v,s)\}}\prod_{i=1}^n f_0(y_i) dy_1,\ldots, dy_n,
	\]
	Taylor expansion of $\kappa$ for small $\delta_1,\ldots, \delta_i$ gives
	\begin{align*}
	f^*_{V,S}(v,s; \beta, \gamma, \delta)= &\; \exp\Bigl(v\beta + s^\T\gamma - \sum_{j=1}^m\sum_{i\in\mathcal{I}_j}\kappa(t_i\beta + \gamma_j + \delta_i)\Bigr)\exp\Bigl(\sum_{j=1}^m\sum_{i\in\mathcal{I}_j}\delta_iy_i\Bigr)h_0(v,s),\\
	=& \;  f_{V,S}(v,s; \beta, \gamma)\exp\Bigl\{\sum_{j=1}^m\sum_{i\in\mathcal{I}_j} \Bigl(\delta_iy_i - \sum_{k=1}^\infty \frac{\delta_i^k}{k!} \kappa^{(k)}(t_i\beta + \gamma_j )\Bigr)\Bigr\},
	\end{align*}
	where $f_{V,S}(v,s; \beta, \gamma)$ is the joint density function of $V$ and $S$ under the piecewise constant approximation.
	
	Within the notional family, the conditional density function $f_{V|S}(v \mid s; \beta, \gamma)$ is obtained from the joint density function on division by the marginal
	\[
	f_S(s; \beta, \gamma) = \exp\Bigl(s^\T\gamma - \sum_{j=1}^m\sum_{i\in\mathcal{I}_j}\kappa(t_i\beta + \gamma_j)\Bigr)\int \exp(v\beta)h_0(v,s) dv,
	\]
	giving
	\[
	\frac{f^*_{V|S_1,\ldots, S_m}(v\mid s_1^o,\ldots, s_m^o;\beta, h)}{f_{V|S_1,\ldots, S_m}(v\mid s_1^o,\ldots, s_m^o;\beta)}=\exp\biggl\{\sum_{i=1}^n\delta_iy_i - \sum_{k=1}^\infty\sum_{j=1}^m \sum_{i\in \mathcal{I}_j}\frac{\delta_{i}^{k}\kappa^{(k)}(t_i\beta + \gamma_j)}{k!}\biggr\},
	\]
	as required. The same strategy proves \eqref{eqLemmaPart}.
\end{proof}

\begin{proof}[Proof of Proposition \ref{propApproxError2}]
	On treating $Y_1,\ldots,Y_n$ as random variables in the evaluation point $v(Y_1,\ldots,Y_n)$ of the conditional density functions $f^*_{V|S_1,\ldots, S_m}$ and $f_{V|S_1,\ldots, S_m}$, and taking log expectations under $f^*_Y$ of both sides of the resulting version of \eqref{eqLemma}, gives
	\[
	\log\EE\biggl[\frac{f^*_{V|S_1,\ldots, S_m}(v(Y_1,\ldots, Y_n)\mid s_1^o,\ldots, s_m^o;\beta, h)}{f_{V|S_1,\ldots, S_m}(v(Y_1,\ldots, Y_n)\mid s_1^o,\ldots, s_m^o;\beta)}\biggr] = \sum_{i=1}^n K_{Y_i}(\delta) -  \sum_{k=1}^\infty\sum_{j=1}^m \sum_{i\in \mathcal{I}_j}\frac{\delta_{i}^{k}\kappa^{(k)}(t_i\beta + \gamma_j)}{k!},
	\]
	where $K_{Y_i}(\delta_i)= \kappa(t_i\beta + h(x_i) + \delta_i) - \kappa(t_i\beta + h(x_i))$ is the cumulant generating function of $Y_i$ at $\delta_i$. A Taylor series expansion for small $\delta_i$ gives \eqref{eqApproxError}.
\end{proof}

\begin{proof}[Proof of Proposition \ref{propReverse}]
The conditional log-likelihood function for $\beta$ under the piecewise constant model satisfies
\[
\sum_{j=1}^m \biggl(\log f_{V_j|S_j}(v_j \mid s_j^o ; \beta) - \log f^*_{V_j|S_j}(v_j \mid s_j^o ; \beta, h)\biggr) = \sum_{j=1}^m \log \biggl(\frac{f_{V_j|S_j}(v_j \mid s_j^o ; \beta)}{f^*_{V_j|S_j}(v_j \mid s_j^o ; \beta, h)}\biggr).
\]
On viewing $(Y_i)_{i\in\mathcal{I}_j}$ as random variables in \eqref{eqLemmaPart} and taking expectations under \eqref{eqGeneralForm} gives
\[
\EE \sum_{j=1}^m \log \biggl(\frac{f_{V_j|S_j}(v_j \mid s_j^o ; \beta)}{f^*_{V_j|S_j}(v_j \mid s_j^o ; \beta, h)}\biggr) = \sum_{j=1}^m \sum_{i\in \mathcal{I}_j} \Bigl(\sum_{k=1}^\infty \frac{\delta_{i}^{k}\kappa^{(k)}(t_i\beta + \gamma_j)}{k!} - \delta_i \kappa^\prime(t_i\beta + h(x_i))\Bigr).
\]
An application of equation \eqref{eqCumExp} gives the stated result.
\end{proof}

\begin{proof}[Proof of Proposition \ref{propReplicates}]
Since
\[
\frac{U_j^*(\beta)}{U_j(\beta)} = \frac{\int_{-\infty}^{V_j} f_{V_j|S_j}^*(v \mid s_j^o ; \beta, h)d\mu(v)}{\int_{-\infty}^{V_j} f_{V_j|S_j}(v \mid s_j^o ; \beta)d\mu(v)},
\]
Lemma \ref{lemmaApproxError} gives 
\[
A_m(\beta) - A_m^*(\beta) 
= 2\sum_{j=1}^m \sum_{j\in \mathcal{I}_j} \biggl(\sum_{k=1}^{\infty}\frac{\delta_{i}^k \kappa^{(k)}(t_i\beta + \gamma_j)}{k!} - \delta_i Y_i\biggr)
\]
by direct calculation, from which the statements about mean and variance follow by \eqref{eqCumExp}.
\end{proof}

	\section{Proof of Proposition \ref{propPairs}}

\begin{proof}
	For notational convenience, introduce $V^{(k)}_j$ for the, perhaps notional, random variable relevant for exploiting a postulated model with Structure $k$. This is $Z_j$ and $Z_j(\psi_0)$ for $k=3$ and $k=4$ respectively and is the notional random variable $Y_{j1}(S_j)$ and $Y_{j1}(S_j(\psi_0))$ for $k=1$ and $k=2$. These are not functions of $S_j$ in the conventional sense, and are notional in the sense that they typically cannot be expressed in terms of the original random variables, but once $S_j=s_j^o$ or $S_j(\psi_0)=s_j^o(\psi_0)$ is observed, $Y_{j1}(S_j)$ and $Y_{j1}(S_j(\psi_0))$ have the conditional distribution of $Y_{j1}$ given $S_j=s_j^o$ or $S_j(\psi_0)=s_j^o(\psi_0)$. 
	
	Temporarily dropping subscripts and superscripts, let $F_V(v;\psi_0)$ be the distribution function of $V$ under the assumption that the postulated model is true at $\psi_0$. Let $G_V^*(v)$ be the true distribution of $V$. Thus, if $G_V^*$ belongs to the family $\mathcal{F}_V:=\{F_V(\cdot ; \psi), \; \psi \in \Psi\}$, then there exists a $\psi^*$ such that $G_V^*(v)=F_V(v;\psi^*)$.
	
	One direction is by direct calculation: if $\psi_0=\psi^*$, then the distribution of $F_V(V;\psi_0)$ is uniform. For the converse direction, suppose for a contradiction that $F_{V}(V;\psi_0)$ is uniformly distributed but that $G_V^*(v)\neq F_V(v;\psi_0)$. By uniformity,
	\begin{equation}\label{eqUnif}
	u= \text{pr}\bigl(F_{V}(V; \psi_0) \leq u\bigr) 
	\end{equation}
	for all $u\in[0,1]$. First suppose that $G_V^*\in \mathcal{F}_V$. Since $F_{V}(v,\psi)$ is bijective in $v$ for any fixed $\psi$, there exists a $u\in[0,1]$ such that $F_{V}^{-1}(u; \psi_0)\neq F_{V}^{-1}(u; \psi_0')$ for any $\psi_0'\neq \psi_0$. Thus, for any such value of $u$, $F_{V}\bigl(F_{V}^{-1}(u; \psi_0); \psi^*\bigr)\neq F_{V}\bigl(F_{V}^{-1}(u; \psi^*); \psi^*\bigr)$, where the right hand side is equal to $u$ by definition. It follows that $F_{V}\bigl(F_{V}^{-1}(u; \psi_0); \psi^*\bigr) \neq u$, which contradicts \eqref{eqUnif}.
	
	Suppose now that uniformity is achieved and $G_V^*\notin \mathcal{F}_V$. Then $G_{V}^*(F_{V}^{-1}(u; \psi_0))$ replaces  $F_{V}\bigl(F_{V}^{-1}(u; \psi_0); \psi^*\bigr)$ in the argument of the previous paragraph. Equality for all $u$ can only be achieved at the points $\psi_0$ where $G_{V}^*(v)=F_{V}(v; \psi_0)$. But at such points $G_{V}^*$ belongs to $\mathcal{F}_V$, a contradiction.
	
	The remaining question is whether there are distributions over $(Y_{j0},Y_{j1})$ that induce the same distribution $G_{V_j}^*(v)=F_{V_j}(v;\psi_0)$ over $V_j$ that would hold if the postulated model for $(Y_{j0},Y_{j1})$ were true at $\psi_0$, and since the postulated model induces a replication, this does not depend on the index $j$. The corresponding equivalence classes are most easily expressed in terms of density functions $f^*_{Y_{j0},Y_{j1}}$, inducing a density function $g_{V_j}^*=g_V$ on $V_j$, where $g_{V_j}^*$ denotes the true induced density function and $g_V$ denotes the density function under the postulated model, identical for all $j$ by replication inducement under the postulated model. These are the sets $\mathcal{E}^{(k)}(\psi_0)$ from Proposition \ref{propPairs}, made explicit in Appendix \ref{secEquiv}.
\end{proof}

\section{Equivalence classes for Proposition \ref{propPairs}}\label{secEquiv}

Proposition \ref{propPairs} refers to equivalence classes $\mathcal{E}^{(k)}(\psi_0)$ of density functions $f^*_j$ over pairs $(Y_{j0}, Y_{j1})$, within which the approach to achieving internal replication based on Structures \ref{structureSuff1}--\ref{structureTrasform2} has no power to reject the postulated model if the containment $f^*_j\in \mathcal{E}^{(k)}(\psi_0)$ holds for all $j$. In other words, if the true density functions $f_j^*$, $j=1,\ldots, m$ violate the postulated model in the directions of any model belonging to the corresponding equivalence classes $\mathcal{E}^{(k)}(\psi_0)$, then this will not be detected at any sample size. With $|J|$ the absolute Jacobian determinant corresponding to the transformation from $(Y_{j0},Y_{j1})$ to $(S_j, Y_{j1})$ for $k=1$, to $(S_j(\psi_0), Y_{j1})$ for $k=2$, to $(Z_j, Y_{j1})$ for $k=3$ and to $(Z_j(\psi_0),Y_{j1})$ for $k=4$, the equivalence classes are:
\begin{align*}
\mathcal{E}^{(1)}(\psi_0)
= {} & \Bigl\{\, f :
\frac{
	f\bigl(y_0(s_j^o,y_1), y_1\bigr)\, |J|
}{
	\int_{\mathcal{Y}_1(s_j)} 
	f\bigl(y_0(s_j^o,y_1), y_1\bigr)\, |J| \, dy_1
}
= g_{Y_{j1} \mid S_j}(y_1 \mid s_j^o ; \psi_0)
\,\Bigr\}, \\[0.75em]
\mathcal{E}^{(2)}(\psi_0)
= {} & \Bigl\{\, f :
\frac{
	f\bigl(y_0(s_j^o(\psi_0), y_1), y_1\bigr)\, |J|
}{
	\int_{\mathcal{Y}_1(s_j^o(\psi_0))}
	f\bigl(y_0(s_j^o(\psi_0), y_1), y_1\bigr)\, |J| \, dy_1
}
= g_{Y_{j1} \mid S_j(\psi_0)}\bigl(y_1 \mid s_j^o(\psi_0)\bigr)
\,\Bigr\}, \\[0.75em]
\mathcal{E}^{(3)}(\psi_0)
= {} & \Bigl\{\, f :
\int_{\mathcal{Y}_1}
f\bigl(y_0(z, y_1), y_1\bigr)\, |J| \, dy_1
= g_{Z_j}(z ; \psi_0)
\,\Bigr\}, \\[0.75em]
\mathcal{E}^{(4)}(\psi_0)
= {} & \Bigl\{\, f :
\int_{\mathcal{Y}_1}
f\bigl(y_0(z(\psi_0), y_1), y_1\bigr)\, |J| \, dy_1
= g_{Z_j(\psi_0)}\bigl(z(\psi_0)\bigr)
\,\Bigr\}.
\end{align*}
The density functions on the right hand side of each expression are, in the notation of the proof of Proposition \ref{propPairs}, the density functions $g_V$ of the induced random variables under the postulated model. That the equivalence classes $\mathcal{E}^{(k)}(\psi_0)$, $k=1,\ldots,4$, cannot depend on $j$ is due to the induced replication under the postulated model. 

The above equivalence relations are in terms of generic expressions and usually do not specify the simplest route to calculation in particular cases. For instance, the distribution of $Z_j$ from Example \ref{exIntroCont} can typically be calculated directly without explicit preliminary transformation from $(Y_{j0},Y_{j1})$ to $(Z_j,Y_{j1})$. 

\section{Derivations for \S \ref{secPowerPairs}}\label{secPowerDerivations}

\subsection{Local consistency under the Weibull model}\label{secWeibullDeriv}

Maximisation of the log-likelihood function based on \eqref{eqDistExpMult} when the true distribution is given by \eqref{eqWeibullRatioDist} produces a maximum likelihood estimator $\hat\psi\rightarrow_p \psi_0^*$, where $\psi_0^*$ is the value of $\psi$ that maximises the expected log-likelihood function
\begin{equation}\label{eqKL}
\psi_0^* = \argmax_{\psi\in \RR^+} \biggl(\log\psi - 2\int_{0}^{\infty}\log(1+ \psi z)\frac{\psi^*\varsigma z^{\varsigma-1}dz}{(1+ \psi^* z^{\varsigma})^2} \biggr).
\end{equation}
This solution satisfies
\begin{equation}\label{eqSolution}
\frac{1}{\psi_0^*}=2\psi^*\varsigma\int_{0}^{\infty}\frac{ z^{\varsigma}dz}{(1+ \psi^* z^{\varsigma})^2(1+\psi_0^* z)}.
\end{equation}
Since we are interested in the behaviour near the point of intersection of the two models $\varsigma=1$, write $\varsigma=1+\varepsilon$ and expand the integral for small $\varepsilon$, giving
\[
\frac{1}{\psi_0^*}=2\psi^*(1+\varepsilon)\biggl(\int_{0}^{\infty}\frac{ z dz}{(1+ \psi^* z)^2(1+\psi_0^* z)}+O(\varepsilon)\biggr)
\]
On expanding the integral in partial fractions it follows that, to first order in $\varepsilon$, $\psi_0^*$ solves
\[
\frac{\psi^*}{\psi_0^*}=2\psi^*(1+\varepsilon)\biggl(\frac{\psi^*\log(\psi^*/\psi_0^*)+\psi_0^*-\psi^*}{(\psi^*-\psi_0^*)^2}\biggr).
\]
Equivalently, $x=\psi^*/\psi_0^*$ solves
\begin{equation}\label{eqSolveX}
(x-1)^2=2(1+\varepsilon)x(\log x + 1/x - 1).
\end{equation}
The left and right hands sides are both convex with a unique minimum at $x=1$ for any $\varepsilon$. Thus, to first order in $\varepsilon$, the maximum likelihood solution $\psi_0^*$ under the postulated model is equal to the true value $\psi^*$.

\subsection{A power analysis for Example \ref{exampleTrueWeibull} with estimated $\psi$}\label{secWeibullEstPsi}

From \S \ref{secWeibullDeriv}, the probability limit of the maximum likelihood estimator based on \eqref{eqDistExpMult} is $\psi_0^*=\psi^* + O(\varsigma-1)$. The linearisation argument from \eqref{eqUhat}, expression \eqref{eqMono} and a Taylor expansion for small $\varepsilon=\varsigma-1$, gives
\begin{equation}\label{eqUMisspec}
\text{pr}(U_j(\hat\psi)\leq u) \rightarrow \text{pr}(U_j(\psi_0^*)\leq u) = \frac{\psi^{*(1-\varsigma)} u^\varsigma}{(1-u)^{\varsigma}+\psi^{*(1-\varsigma)} u^{\varsigma}} + O(\varsigma-1).
\end{equation}
This shows that $\hat{A}_m$ in \eqref{eqHatVersion} has power to detect departures from the model even though the maximum likelihood estimator of $\psi^*$, based on the transformed random variables $Z_j$, is consistent modulo $O(\varsigma-1)$ under the misspecified model. In particular, from \eqref{eqIntParts}, $\EE(R_j)=1$ only when $\varsigma=1$, the point at which the Weibull model collapses to the exponential model. Figure \ref{figHeatmap} plots the logarithm of both $\EE(R_j)$ and $\VV(R_j)$, obtained by numerical integration as a function of $\varsigma$ and $\psi^*$. This shows that \eqref{eqHatVersion} is most sensitive to departures from the postulated model in the direction of $\varsigma<1$, with neighbourhoods of $\varsigma=1$ and $\psi^*$ being difficult to detect, uniformly over the values of the other variable. There is power to detect departures from the model when $(\varsigma,\psi^*)$ take values in the upper right quadrant of the reported range. 

\begin{figure}
	\centering
	\includegraphics[trim=0.6in 2.9in 0.6in 3.3in, clip, height=0.385\textwidth]{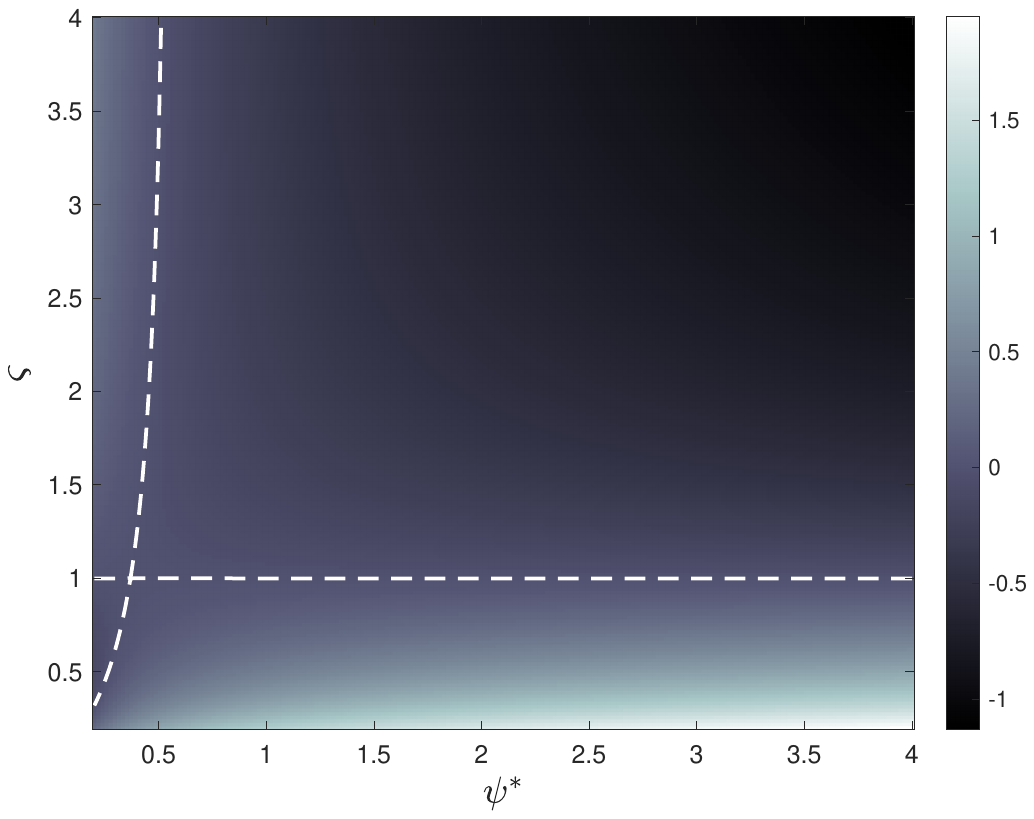}
	\includegraphics[trim=0.6in 2.9in 0.6in 3.3in, clip, height=0.385\textwidth]{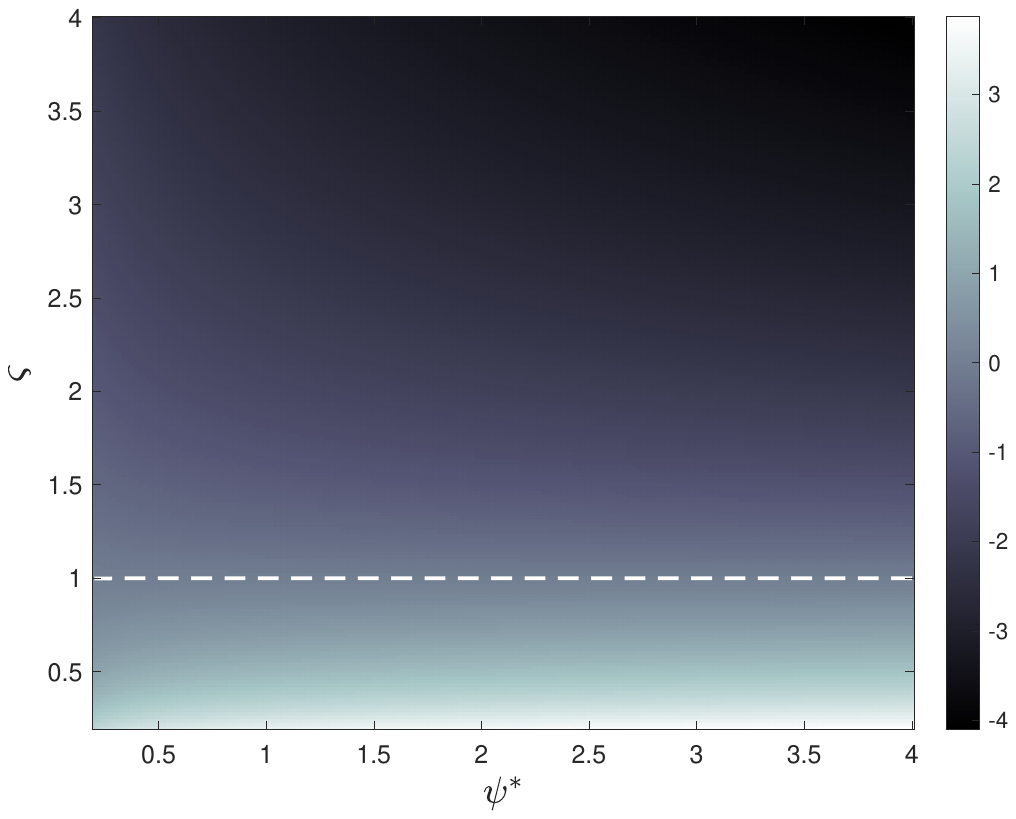}
	\caption{Logarithm of $\EE(R_j)$ and $\VV(R_j)$ calculated from \eqref{eqUMisspec} and \eqref{eqIntParts}, showing how the sensitivity to detect model misspecification varies with $\varsigma$ and $\psi^*$. Dashed lines indicate regions of the parameter space for the true Weibull model where $\EE(R_j)$ and $\VV(R_j)$ coincide with their  putative values under the hypothesised exponential model. 
		\label{figHeatmap}}
\end{figure}	

\subsection{Power analysis for another perturbative case}\label{secAdditivePower}

\begin{example}\label{exampleTrueAdditive}
	Another type of departure from the proportional exponential model of Example \ref{exIntro} is in the direction of the additive-rates exponential model of Example \ref{exampleAdditiveRate}, with treatment effect $\Delta^*$. In the additive model, the nuisance parameters are eliminated by conditioning on the pairwise sums, while for the multiplicative model, conditioning is on $Y_{j0} + Y_{j1}\psi_0$ for postulated values of $\psi$. Since there is no true value $\psi^*$ when the model is erroneous, this second conditioning argument is ineffective at eliminating the nuisance parameters except at $\psi_0=1$. 
	
	Since Structure \ref{structureTrasform1} leads to an identical inferential procedure as Structure \ref{structureSuff2} under the multiplicative model, interest lies in the behaviour of $Z_j=Y_{j1}/Y_{j0}$ under the additive model, which has distribution function
	\begin{equation}\label{eqDistZAdd}
	F_{Z_j}(z) =    \text{pr}(Z_j \leq z) = \frac{(\gamma_j + \Delta^*)z}{{\gamma_j + (\gamma_j + \Delta^*)z}}. 
	\end{equation}
	Since \eqref{eqDistExpMult} and \eqref{eqDistZAdd} agree at the null effect value $\Delta^*=0$, which corresponds to $\psi=1$, we expect greater difficulty in detecting inadequacy of the multiplicative model for small values of $\Delta^*$. A more exact analysis of power echoing the general discussion of \S \ref{secTwo} requires calculation of the density function of $U_j(\psi_0)=\psi_0 Z_j/(1+ \psi_0 Z_j)$ when the distribution of $Z_j$ is specified by \eqref{eqDistZAdd}. By monotonicity,
	\begin{align}\label{eqMono2}
	\text{pr}(U_j(\psi_0) \leq u)= \text{pr}\Bigl(Z_j \leq \frac{u}{\psi_0(1-u)} \Bigr) = \frac{(\gamma_j + \Delta^*)u}{\gamma_j\psi_0(1-u) + (\gamma_j + \Delta^*)u}, \quad 0\leq u \leq 1,
	\end{align}
	with corresponding density function
	\[
	f_{U}(u) = \frac{\gamma_j (\gamma_j+\Delta^*) \psi_0}{(\gamma_j\psi_0 (1-u) + (\gamma_j+\Delta^*) u)^2}, \quad 0\leq u \leq 1.
	\]
	This is not of the form \eqref{eqPara}, even approximately for $\Delta^*$ close to 0, but exact calculation via \eqref{eqIntParts} shows that the expectation of $R_j=-\log U_j(\psi_0)$ is
	\[
	\EE(R_j) = \frac{\gamma_j+\Delta^*}{\gamma_j+\Delta^* - \gamma_j\psi_0}\log \Bigl(\frac{\gamma_j+\Delta^*}{\gamma_j\psi_0}\Bigr) = \frac{\eta_j \log \eta_j}{\eta_j -1}, \quad \gamma_j + \Delta^*>0,
	\]
	where $\eta_j = (\gamma_j + \Delta^*)/\gamma_j\psi_0$. This is monotonically increasing in $\eta_j>0$ and crosses $\EE(R_j)=1$ at $\eta_j=1$, which corresponds to $\Delta^*>\gamma_j(\psi_0-1)$. Thus,  sensitivity in the right tail tends to increase with $\Delta^*$, and values $\psi_0\leq 1$ are more easily detected as erroneous than $\psi_0>1$ when $\Delta^*>0$, confirming intuition. The criterion $\Delta^*>\gamma_j(\psi_0-1)$ does, however, suggest that there are larger values of postulated $\psi_0$ that make the erroneous postulated model difficult to refute on the basis of the aggregation criterion \eqref{eqConfSet}, in spite of violation of standard uniformity; this is supported by numerical checks. It seems necessary in this example, therefore, to use a statistic based on $\hat\psi$, e.g.~\eqref{eqHatVersion}, echoing the conclusion from the previous example. 
	
	Appendix \ref{secAdditiveDeriv} presents a local asymptotic expansion of the maximum likelihood solution in a neighbourhood of the point of intersection  $\Delta^*=0$ of the two models showing, as expected, that there is little sensitivity to departures from the postulated multiplicative model in the direction of the additive model at small $\Delta^*$. \hfill$\square$
\end{example}

\subsection{Local asymptotic analysis for the additive exponential model}\label{secAdditiveDeriv}

The limiting maximum likelihood solution is in this case more complicated in view of the nuisance parameters, but by the law of large numbers the average of the log-likelihood contributions converges to the average of the expected log-likelihood contributions, thus
\[
\psi_0^* = \lim_{m\rightarrow \infty}\argmax_{\psi\in \RR^+} \biggl(\log\psi - \frac{2}{m}\sum_{j=1}^m\gamma_j(\gamma_j + \Delta^*)\int_{0}^{\infty}\frac{\log(1+ \psi z) dz}{(\gamma_j+ (\gamma_j + \Delta^*)z)^2} \biggr).
\]
Differentiation shows that $\psi_0^*$ solves
\[
\frac{1}{\psi_0^*} = \lim_{m\rightarrow \infty} \frac{2}{m}\sum_{j=1}^m\gamma_j(\gamma_j + \Delta^*)\int_{0}^{\infty}\frac{z dz}{(1+\psi_0^* z)(\gamma_j + (\gamma_j+\Delta^*)z)^2}.
\]
Analogously to Example \ref{exampleTrueWeibull} we can expand the integrand around $\Delta^*=0$, which corresponds to the point at which the two models intersect, giving, to first order
\[
\frac{1}{\psi_0^*} = \biggl(\frac{\psi_0^* - 1 - \log \psi_0^*}{(\psi_0^* - 1)^2} + O(\Delta^*)\biggr)\lim_{m\rightarrow \infty} \frac{2}{m}\sum_{j=1}^m\frac{(\gamma_j + \Delta^*)}{\gamma_j}.
\]
To the same order, this is of almost identical form to \eqref{eqSolveX} with the term $1+\varepsilon$ replaced by the limiting average $a$ of $(\gamma_j + \Delta^*)/\gamma_j$. The first-order solution under the notional asymptotic regime $\Delta^*\rightarrow 0$ is thus the logical one $\psi_0^*=1$ for all $a$, $\psi^*=1$ and $\Delta^*=0$ corresponding to the null treatment effect and the point at which the two models intersect. Thus, intuitively, there is little sensitivity to departures from the postulated multiplicative model in the direction of the additive model at small $\Delta^*$.

\section{Numerical checks}\label{secSim}

\subsection{Matched pairs}\label{secSimMP}

Data were generated from a joint model for matched pairs $(Y_{j0},Y_{j1})_{j=1}^m$ in which an additive treatment effect $\Delta^*$ operates on the rate scale for Weibull distributed outcomes with baseline rate parameters $(\gamma_j)_{j=1}^m$ and shape parameter $\varsigma$, where $(\gamma_j)_{j=1}^m$ were generated from a standard uniform distribution. The assumed model is exponential with a multiplicative rate parameter as in Example \ref{exIntro}. This notional parameter was estimated by maximum likelihood based on the density function of the ratios $Z_j=Y_{j1}/Y_{j0}$ which, from \eqref{eqDistExpMult} is $\psi/(1+\psi z)^2$. Fisher's statistic $2\sum_{j=1}^m \log U_j(\hat\psi)$ was then constructed based on \eqref{eqDistExpMult} with $\psi^*$ replaced by $\hat\psi$ and $z$ replaced by $Z_j$; Table \ref{tabAdditive} reports the results. At $(\Delta^*, \varsigma)=(0,1)$, the true model intersects with the postulated model with null treatment effect, captured by value $\psi=1$. Thus, the proportion of rejected tests of size $\alpha$ would be $\alpha$ by construction if $\psi$ was perfectly estimated. Table \ref{tabAdditive} reports a beneficial under-rejection rate in that case. Elsewhere in the parameter space, power to detect departures from the model increases to 1 with increasing $m$ except when $\varsigma=1$, where the true Weibull distribution collapses to the exponential and therefore is closer to the point at which the true and postulated models intersect. Both the true distribution and those in the postulated model are in this example constructed from a parameter space whose dimension diverges at the same rate as the sample size $m$, so both models are semiparametric.

We also assessed the approach based on \eqref{eqConfSet} without estimation of the notional parameter $\psi$, computing the proportion of confidence sets for $\psi$ that were empty. If the parameter space of permissible values is constrained to some reasonable interval such as $[0,4]$, then such an approach tends to be more powerful for the sample sizes reported, but a larger range, say $[0,10]$ requires a much larger sample size in order for the confidence set to be empty with high probability; numerical results are not reported.

{\footnotesize{

		\begin{table}[h!]
			
			\begin{center}
				\begin{tabular}{|c|c|cccccc|}
					\hline
					Direction     &	\multirow{2}{*}{($\Delta^*$, $\varsigma$) }     &  \multicolumn{6}{|c|}{Square root of number of pairs $m$} \\
					\cline{3-8}
					of sensitivity &				& $5$ & $8$ & $11$ & $14$ & $17$ & $20$  \\
					\hline				
					left & (0, 0.5)   & 0.657 & 0.977 & 1 & 1 & 1 & 1  \\
					right & (0, 0.5)   & 0.640 & 0.983 & 1 & 1 & 1 & 1  \\
					\hline
					\textbf{left} & \textbf{(0, 1)}   & \textbf{0.002} & \textbf{0.001} & \textbf{0} & \textbf{0} & \textbf{0} & \textbf{0}  \\
					\textbf{right} & \textbf{(0, 1)}   & \textbf{0.001} & \textbf{0.001} & \textbf{0} & \textbf{0.001} & \textbf{0} & \textbf{0}  \\
					\hline
					left & (0, 2)   & 0 & 0.221 & 0.994 & 1 & 1 & 1  \\
					right & (0, 2)   & 0 & 0.232 & 0.997 & 1 & 1 & 1  \\
					left & (1, 0.5)   &  0.717 & 0.982 & 1 & 1 & 1 & 1  \\
					right & (1, 0.5)   & 0.685 & 0.990 & 1 & 1 & 1 & 1  \\
					left & (1, 1)   & 0.008 & 0.015 & 0.020 & 0.042 & 0.075 & 0.134  \\
					right & (1, 1)   & 0.002 & 0.003 & 0 & 0.005 & 0.007 & 0.010  \\
					left & (1, 2)   & 0 & 0 & 0.105 & 0.420 & 0.774 & 0.951  \\
					right & (1, 2)   & 0 & 0.005 & 0.647 & 0.992 & 1 & 1  \\
					left & (2, 0.5)   & 0.731 & 0.983 & 1 & 1 & 1 & 1  \\
					right & (2, 0.5)   & 0.696 & 0.992 & 1 & 1 & 1 & 1  \\
					left & (2, 1)   & 0.010 & 0.025 & 0.030 & 0.073 & 0.134 & 0.236 \\
					right & (2, 1)   & 0.002 & 0.003 & 0.002 & 0.006 & 0.012 & 0.023  \\
					left & (2, 2)  & 0 & 0 & 0.041 & 0.215 & 0.555 & 0.825  \\
					right & (2, 2)   & 0 & 0.004 & 0.482 & 0.979 & 1 & 1  \\
					\hline
				\end{tabular}
				\medskip
				\caption{Weibull generating distribution with additive treatment parameter $\Delta^*$ and shape $\varsigma$. The assumed model is exponential with multiplicative treatment parameter $\psi$. Proportion of 1000 simulation runs in which the 5\% test based on \eqref{eqHatVersion} rejects the postulated model. The point of model intersection $(\Delta^*,\varsigma)=(0,1)$ is highlighted. 	\label{tabAdditive}}
				\vspace{-0.5cm}
			\end{center}
			
		\end{table}
}}

{\footnotesize{

		\begin{table}[h!]
			
			\begin{center}
				\begin{tabular}{|c|c|cccccc|}
					\hline
					Direction     &	\multirow{2}{*}{($\psi^*$, $\varsigma$) }     &  \multicolumn{6}{|c|}{Square root of number of pairs $m$} \\
					\cline{3-8}
					of sensitivity &				& $5$ & $8$ & $11$ & $14$ & $17$ & $20$  \\
					\hline				
					left & (0.5, 0.5)   & 0.442 & 0.589 & 0.763 & 0.859 & 0.926 & 0.964 \\
					right & (0.5, 0.5)   & 0.846 & 0.995 & 1 & 1 & 1 & 1  \\
					left & (0.5, 1)   & 0.231 & 0.623 & 0.907 & 0.991 & 0.998 & 1 \\
					right & (0.5, 1)   & 0.239 & 0.669 & 0.942 & 0.998 & 1 & 1  \\
					left & (0.5, 2)   & 0.875 & 1 & 1 & 1 & 1 & 1 \\
					right & (0.5, 2)   & 0.002 & 0.035 & 0.182 & 0.536 & 0.817 & 0.939 \\
					left & (1, 0.5)   & 0.708 & 0.925 & 0.990 & 1 & 1 & 1  \\
					right & (1, 0.5)   & 0.565 & 0.881 & 0.984 & 1 & 1 & 1  \\
					\hline
					\textbf{left} & \textbf{(1, 1)}  & \textbf{0.079} & \textbf{0.068} & \textbf{0.082} & \textbf{0.059} & \textbf{0.070} & \textbf{0.058} \\
					\textbf{right} & \textbf{(1, 1)}   & \textbf{0.061} & \textbf{0.075} & \textbf{0.073} & \textbf{0.060} & \textbf{0.065} & \textbf{0.054} \\
					\hline
					left & (1, 2)   & 0.046 & 0.399 & 0.826 & 0.985 & 0.999 & 1 \\
					right & (1, 2)  & 0.073 & 0.450 & 0.894 & 0.991 & 1 & 1 \\
					left & (2, 0.5)  & 0.896 & 0.999 & 1 & 1 & 1 & 1  \\
					right & (2, 0.5)   & 0.329 & 0.466 & 0.606 & 0.740 & 0.873 & 0.936 \\
					left & (2, 1)   & 0.253 & 0.626 & 0.925 & 0.993 & 1 & 1 \\
					right & (2, 1)  & 0.252 & 0.620 & 0.922 & 0.991 & 1 & 1 \\
					left & (2, 2)  & 0.011 & 0.042 & 0.192 & 0.515 & 0.763 & 0.901 \\
					right & (2, 2)   & 0.889 & 1 & 1 & 1 & 1 & 1 \\
					\hline
				\end{tabular}
				\medskip
				\caption{Weibull generating distribution with multiplicative treatment prameter $\psi^*$ and shape $\varsigma$. The assumed model is exponential with additive treatment parameter $\Delta$. Proportion of 1000 simulation runs in which the 5\% test based on \eqref{eqHatVersion} rejects the postulated model. The point of model intersection $(\psi^*,\xi)=(1,1)$ is highlighted. 	\label{tabMultip}}
				\vspace{-0.5cm}
			\end{center}
			
		\end{table}
}}

Table \ref{tabMultip} reverses the roles of the two types of models, taking the baseline distribution of $Y_{j0}$ to be the same as before but obtaining the distribution of $Y_{j1}$ through multiplication of the rate parameter $\gamma_j$ by a value $\psi^*$. The postulated model, by contrast, is the exponential additive rates model of Example \ref{exampleAdditiveRate}. There is high power to detect the erroneous model, and at the point $(\psi^*, \varsigma)=(1,1)$ at which the true and postulated models coincide, the nominal level is recovered asymptotically.

\subsection{Time-dependent Poisson process}\label{secSimPoisson}

Simulation of an inhomogeneous Poisson process can be performed most simply using the argument of \citet[][pp.~27--28]{CL1966} that on the transformed time scale $\tau_i(t)=\int_0^t \lambda_i(u)du$ the process is Poisson of unit constant rate. It follows that the sequence  $T_{i1},\ldots, T_{im_i}$ of ordered  event times for individual $i$ satisfy $\tau_i(T_{i k_i})=\sum_{j=1}^{k_i} E_{ij}$, where $E_{ij}$ are independent unit exponential random variables. For $\lambda_i(t)=e^{\gamma_i + \beta t}$, we have $\tau_i(t)= e^{\gamma_i}(e^{\beta t}-1)/\beta$, and the event times are distributed as
\begin{equation}\label{eqSimT}
T_{ik_i}\stackrel{d}{=} \tau_i^{-1}\Bigl(\sum_{j=1}^{k_i} E_{ij}\Bigr) =\log\Bigl(1+(\beta/e^{\gamma_i}) \sum_{j=1}^{k_i} E_{ij}\Bigr)/\beta,
\end{equation}
the limit as $\beta\rightarrow 0$ being $T_{ik_i}=e^{-\gamma_i}\sum E_{ij}$ as expected. To check that the nominal error rates are recovered under the postulated model, we used this generating process with $\gamma_i$ drawn from a standard normal distribution, generating unit exponential random variables to use in \eqref{eqSimT} until their sum exceeded the endpoint of the transformed timescale $\tau_i(t_0)$  with $t_0=5$. Since we wish to assess cases where the model is misspecified, we also simulated using the power-law intensity function $\lambda_i(t)=e^{\gamma_i}t^\rho$ with $\rho\in\{-0.5, 0, 0.5, 1\}$. We generated
\begin{equation}\label{eqSimT2}
T_{ik_i}\stackrel{d}{=} \tau_i^{-1}\Bigl(\sum_{j=1}^{k_i} E_{ij}\Bigr) =\biggl(\frac{(\rho+1)\sum_{j=1}^{k_i} E_{ij}}{e^{\gamma_i}}\biggr)^{1/(\rho+1)}
\end{equation}
until the corresponding sum of unit exponentials exceeded  $\tau_i(t_0)=e^{\gamma_i}t_0^{\rho+1}/(\rho+1)$. At $\rho=\beta=0$ the true and postulated models intersect, but since the likelihood function \eqref{eqConditionalPDF} is indeterminate there due to division by zero, we do not necessarily expect to recover the nominal level. 

From $n$ individuals we estimated $\beta$ in the postulated model by maximising the conditional log-likelihood function
\begin{equation}\label{eqCondLL}
\ell(\beta) = \sum_{i=1}^n  m_i\log\Bigl(\frac{\beta}{e^{\beta t_0} - 1}\Bigr) + \beta \sum_{i=1}^n\sum_{j=1}^{m_i}t_{ij}
\end{equation}
based on \eqref{eqConditionalPDF}. The model was subsequently assessed using the statistic $-2\sum_{i=1}^n \log U_i(\hat\beta)$ where $U_i(\hat\beta)=F_{S_i|M_i}(S_i \cond m_i; \hat\beta)$ and the conditional distribution was approximated either through Monte Carlo sampling using 1000 replicates under the postulated model if $m_i$ was below 40, or by a normal approximation to the distribution of the sum under the postulated model if $m_i$ exceeded 40. An analogous statistic was computed with $\hat\beta$ replaced by its true value. The results are reported in Table \ref{tabCorrect} for data generated under the postulated model and in Table \ref{tabPoisson} for data generated under the power-law model.

Table \ref{tabCorrect} shows that the procedure rarely rejects the true model when the conditional maximum likelihood estimate is used in the construction of each $U_i(\hat\beta)$. When the true value of $\beta$ is used, it attains a rejection rate close to the nominal level, the discrepancy being primarily due to Monte Carlo sampling error in the approximation to $F_{S_i|M_i}(S_i \cond m_i; \hat\beta)$. Table \ref{tabPoisson} of the main paper reports the power to detect departures from the postulated model.

{\footnotesize{

		\begin{table}[h!]
			
			\begin{center}
				\begin{tabular}{|c|c|c|cccccc|}
					\hline
					Direction     &	\multirow{2}{*}{$\beta$} &  \multirow{2}{*}{estimated/true}  &  \multicolumn{6}{|c|}{Square root of number of individuals $n$} \\
					\cline{4-9}
					of sensitivity &		&		& $3$ & $4$ & $5$ & $6$ & $7$ & $8$  \\
					\hline				
					left & 0  & true & 0.05 & 0.08 & 0.02 & 0.07 & 0.05 & 0.06 \\
					right & 0 & true & 0.05 & 0.07 & 0.04 & 0.05 & 0.02 & 0.05 \\
					left & 0 & estimated & 0.03 & 0.04 & 0 & 0.03 & 0.02 & 0.05 \\
					right & 0 & estimated & 0.03 & 0.03 & 0.02 & 0.03 & 0.01 & 0.03 \\
					left & 1  & true & 0.02 & 0.07 & 0.06 & 0.04 & 0.07 & 0.09 \\
					right & 1 & true & 0.03 & 0.04 & 0.03 & 0.03 & 0.05 & 0.05 \\
					left & 1  & estimated & 0.01 & 0.01 & 0 & 0 & 0 & 0 \\
					right & 1 & estimated & 0 & 0 & 0.01 & 0 & 0 & 0 \\
					left & 2  & true & 0.06 & 0.03 & 0.06 & 0.05 & 0.06 & 0.05 \\
					right & 2 & true & 0.06 & 0.04 & 0.04 & 0.07 & 0.03 & 0.04 \\
					left & 2  & estimated & 0 & 0 & 0 & 0 & 0 & 0 \\
					right & 2 & estimated & 0.01 & 0 & 0 & 0.01 & 0 & 0 \\
					\hline
				\end{tabular}
				\medskip
				\caption{The generating process is the inhomogeneous Poisson process with intensity function $\lambda_i(t)=e^{\gamma_i + \beta t}$ as postulated.  Proportion of 200 simulation runs in which the 5\% test based on $F_{S_i|M_i}$ and \eqref{eqHatVersion} rejects the true model. 	\label{tabCorrect}}
				\vspace{-0.5cm}
			\end{center}
			
		\end{table}
}}

\subsection{Confidence sets of regression models}\label{secSimSynth}

This section presents numerical performance of the approach discussed in \S \ref{secSynthetic}. The original procedure of \citet{BRT} was intended for the situation in which the full set of variables contemplated is inordinately large, necessitating some preliminary reduction by  variable screening prior to assessment of low-dimensional subsets of variables. Since, in the present paper, we are illustrating a particular point, we simplify the problem by assuming that the number of starting variables is 15, so that no preliminary reduction is needed. If that were actually the case, there would be no advantage to using anything other than a likelihood-ratio test of every low-dimensional model against the comprehensive model, but it is reassuring to see coverage probabilities for the confidence sets based on $U_j=F_{Z}(Z_j)$ from \eqref{eqFisherDist1915}. These are reported as a function of the number of synthetic replicates $k$ and the sample size $n$. We also report the simulation-average size of the resulting confidence set of models, i.e.~the number of false models that are included in the confidence set on average over simulation runs.

The experiment was conducted as follows. In each of $500$ simulations, the $n$ rows of the $n\times d$ covariate matrix $X$ were drawn from a mean-zero normal distribution with correlation $0.9$ between $s+a$ of the $d$ variables, and correlation zero elsewhere, where $d=15$, $s=5$ is the number of signal variables and $a=3$ is the number of noise variables that are correlated with signal variables. Associated with $X$ is a vector $\theta$ of regression coefficients with entries 1 in the positions corresponding to the signal variables, and zeros elsewhere. The outcome vector was constructed as $Y=X\theta + \varepsilon$, where the entries of $\varepsilon$ were taken as standard normally distributed.

{\footnotesize{

		\begin{table}[h!]
			
			\begin{center}
				\begin{tabular}{|c|c|ccc|}
					\hline
					Direction    &	\multirow{2}{*}{($n$, $k$)}   &  \multirow{2}{*}{Coverage} & \multirow{2}{*}{$\hat{\EE}|\mathcal{M}\backslash \mathcal{S}|$} & \multirow{2}{*}{$\frac{\hat{\EE}(\text{\# false models})}{\text{\# models tested}}$} \\
					of sensitivity & & &  & \\
					\hline				
					left & (100, 4)   & 0.944  & 613 & 0.124 \\
					right & (100, 4)   & 0.954  & 603 & 0.122 \\
					left & (100, 8)   & 0.954  & 508 & 0.103 \\
					right & (100, 8)   & 0.962  & 504 & 0.102  \\
					left & (100, 12)   & 0.964  & 477 & 0.097 \\
					right & (100, 12)   & 0.956  & 476 & 0.096 \\
					left & (100, 16)   & 0.956  & 462 & 0.093 \\
					right & (100, 16)   & 0.968  & 460 & 0.093 \\
					left & (200, 4)   & 0.950  & 270 & 0.055 \\
					right & (200, 4)   & 0.954  & 250 & 0.051 \\
					left & (200, 8)   & 0.956  & 213 & 0.043 \\
					right & (200, 8)   & 0.958  & 206 & 0.041 \\
					left & (200, 12)   & 0.958  & 197 & 0.040 \\
					right & (200, 12)   & 0.958  & 192 & 0.039 \\
					left & (200, 16)   & 0.956  & 190 & 0.038 \\
					right & (200, 16)   & 0.950 & 186 & 0.038 \\
					\hline
				\end{tabular}
				\medskip
				\caption{Simulated coverage probability and expected size of the confidence sets of models from 500 simulations. These were constructed from \eqref{eqHatVersion} based on the distribution \eqref{eqFisherDist1915} of the cosine angles between projected synthetic replicates under each postulated sparse model. \label{tabCSM}}
				\vspace{-0.5cm}
			\end{center}
			
		\end{table}
}}

For every possible model of size $d_0\leq 5$, the corresponding columns $X_0$ of $X$ were extracted. For this postulated model we constructed a basis for the orthogonal projection onto the null space of $X_0$ by taking the $n-d_0$ eigenvectors of $I-X_0(X_0^\T X_0)^{-1}X_0$ corresponding to the unit eigenvalues. Let $V_0$ denote this $n\times (n-d_0)$ matrix of eigenvectors. From the vector of outcomes $Y$, $k\in\{4,8,12,16\}$  synthetic replicates $\tilde Y_1, \ldots, \tilde Y_k$ were generated according to equation (4.3) of \citet{BRT} and the corresponding co-sufficient replicates were constructed as $\tilde Q_j = V_0^\T \tilde Y_j/\|V_0^\T \tilde Y_j\|_2$. The statistics $Z_1,\ldots, Z_m$ were then computed as the $m=k(k-1)/2$ inner products $\langle \tilde Q_j, \tilde Q_i \rangle $, $j\neq i$. When the postulated model is true, these follow the distribution \eqref{eqFisherDist1915} and $U_{j}=F_{Z}(Z_j)$ has a standard uniform distribution. The rejection regions for the model were thus defined in the usual way via \eqref{eqHatVersion}. Table \ref{tabCSM} reports the resulting coverage probabilities of the $0.95$ nominal-coverage confidence sets and the average number of false models in the set from 500 simulations.

\bibliographystyle{plainnat}

\end{appendix}

\end{document}